\documentclass[runningheads]{llncs}

\usepackage{savesym}
\usepackage{amsmath}
\savesymbol{iint}

\usepackage{graphicx}
\usepackage{tikz}
\usetikzlibrary{positioning,shadows,arrows}
\usetikzlibrary{arrows,automata,positioning}
\usepackage{pgf,tikz}

\usepackage{comment}
\usepackage[disable]{todonotes}
\usepackage{shuffle}

\usepackage{wrapfig}
\usepackage{wasysym}
\usepackage{hyperref}

\usepackage{apxproof}

\theoremstyle{plain}
\newtheoremrep{theorem}{Theorem}[section]
 \newtheoremrep{proposition}[theorem]{Proposition}
\newtheoremrep{lemma}[theorem]{Lemma}
 \newtheoremrep{claim}[theorem]{Claim}
 \newtheoremrep{conjecture}[theorem]{Conjecture}
 \newtheoremrep{corollary}[theorem]{Corollary}
 \newtheoremrep{definition}[theorem]{Definition}

\newcommand{\NP}{\textsf{NP}}
\newcommand{\PSPACE}{\textsf{PSPACE}}

\newcommand{\PTIME}{\textsf{P}}

\newcommand{\pref}{\operatorname{Pref}}
\newcommand{\suff}{\operatorname{Suff}}
\newcommand{\factor}{\operatorname{Fact}}

\usepackage{tabularx,environ,amssymb}

\newcounter{problemcounter}
\makeatletter
\newcommand{\problemtitle}[1]{\gdef\@problemtitle{#1}}
\newcommand{\probleminput}[1]{\gdef\@probleminput{#1}}
\newcommand{\problemquestion}[1]{\gdef\@problemquestion{#1}}
\NewEnviron{decproblem}{
  \refstepcounter{problemcounter}
  \problemtitle{}\probleminput{}\problemquestion{}
  \BODY
  \par\addvspace{.5\baselineskip}
  \noindent
  \begin{tabularx}{\textwidth}{@{\hspace{\parindent}} l X c}
    \multicolumn{2}{@{\hspace{\parindent}}l}{\textbf{Decision Problem \theproblemcounter:} \@problemtitle} \\
    \textbf{Input:} & \@probleminput \\
    \textbf{Question:} & \@problemquestion
  \end{tabularx}
  \par\addvspace{.5\baselineskip}
}
\makeatother

\newenvironment{claiminproof}[1]{\medskip\par\noindent\underline{Claim:}\space#1}{}
\newenvironment{claimproof}[1]{\begin{quote}\par\noindent\emph{Proof of the Claim:}\space#1}{[\emph{End, Proof of the Claim}]\end{quote}}

\title{Ideal Separation and General Theorems for Constrained Synchronization and their Application to  Small Constraint Automata} 

\titlerunning{Ideal Separation and Applications to Constrained Synchronization} 

\author{Stefan Hoffmann\orcidID{0000-0002-7866-075X}}
\authorrunning{S.\,Hoffmann} 
\institute{Informatikwissenschaften, FB IV, 
  Universit\"at Trier,  Universitätsring 15, 54296~Trier, Germany, 
  \email{hoffmanns@informatik.uni-trier.de}}

\begin{document}

\maketitle

\begin{abstract}

  In the constrained synchronization problem we
  ask if a given automaton admits a synchronizing
  word coming from a fixed regular constraint language.
  We show that intersecting a given constraint language
  with an ideal language decreases the computational complexity.
  Additionally, we state a theorem giving $\PSPACE$-hardness
  that broadly generalizes previously used constructions
  and a result on how to combine languages
  by concatenation to get polynomial time solvable  constrained synchronization problems. 
  We use these results to give a classification of the complexity
  landscape for small constraint automata of up to three states.
  
  \keywords{Synchronization \and Computational complexity \and Automata theory \and Finite automata} 
\end{abstract}

\section{Introduction}
\label{sec:introduction}


A deterministic semi-automaton is synchronizing if it admits a reset word, i.e., a word which leads to a definite
state, regardless of the starting state. This notion has a wide range of applications, from software testing, circuit synthesis, communication engineering and the like, see~\cite{San2005,Vol2008}.  The famous \v{C}ern\'y conjecture \cite{Cer64}
states that a minimal length synchronizing word, for an $n$-state automaton, has length
at most $(n-1)^2$. 
We refer to the mentioned survey articles for details.

Due to its importance, the notion of synchronization has undergone a range of generalizations and variations
for other automata models.
In some  generalizations, related to partial automata~\cite{Martyugin12}, only certain paths, or input words, are allowed (namely those for which the input automaton is defined). 

In~\cite{Gusev:2012}
the notion of constrained synchronization was 
introduced in connection with a reduction procedure
for synchronizing automata.
The paper \cite{DBLP:conf/mfcs/FernauGHHVW19} introduced the computational problem of constrained 
synchronization. In this problem, we search for a synchronizing word coming from a specific subset of allowed
input sequences. For further motivation and applications we refer to the aforementioned paper \cite{DBLP:conf/mfcs/FernauGHHVW19}.
In this paper, a complete analysis of the complexity landscape when the constraint language is given by 
small partial automata with up to two states and an at most ternary alphabet was done.
It is natural to extend this result to other language classes, or
even to give a complete classification of all the complexity classes that could arise.
For commutative regular constraint languages, a full classification of the realizable
complexities was given in~\cite{DBLP:conf/cocoon/Hoffmann20}.
In~\cite{DBLP:conf/ictcs/Hoffmann20}, it was shown that for polycyclic constraint languages, the problem is always in~$\NP$.

Let us mention that restricting the solution space by a regular language
has also been applied in other areas, for example to topological sorting~\cite{DBLP:conf/icalp/AmarilliP18},
solving word equations~\cite{Diekert98TR,DBLP:journals/iandc/DiekertGH05}, constraint programming~\cite{DBLP:conf/cp/Pesant04}, or
shortest path problems~\cite{DBLP:journals/ipl/Romeuf88}.
The road coloring problem asks for a labelling of a given graph  such that a synchronizing
automaton results. A closely related problem to our problem of constrained synchronization is to restrict the possible labeling(s), and
this problem was investigated in~\cite{DBLP:journals/jcss/VorelR19}.

\medskip

\noindent\textbf{Contribution and Motivation:} In~\cite{DBLP:conf/mfcs/FernauGHHVW19}
a complete classification of the computational complexity
for partial constraint automata with up to two states and an at most ternary alphabet
was given. Additionally, an example of a  a three-state automaton over a binary alphabet realizing
an $\NP$-complete constrained synchronization problem and a three-state automaton
over a binary alphabet admitting a $\PSPACE$-complete problem were given. The question was asked, if, and for what
constraint automata, other complexity classes might arise.
Here, we extend the classification by extending the two-state case to arbitrary alphabets
and giving a complete classification for three-state automata over a binary alphabet. 
It turned out that only $\PSPACE$-complete, or $\NP$-complete, or polynomial time solvable
constrained problems arise.
In~\cite{DBLP:conf/mfcs/FernauGHHVW19}, the analysis for the small constraint automata
were mainly carried out by case analysis. As for larger alphabets and automata this quickly becomes tedious, here we use, and present, new results to lift, extend and combine known results.
Among these are three main theorems, which, when combined, allow many cases to be handled
in an almost mechanical manner. More specifically, the motivation and application
of these theorems is the following.
\begin{enumerate}

\item The \emph{$UV^*W$-Theorem}  describes how to combine languages with concatenation to get polynomial time solvable constrained problems.
 
 \medskip

\item The \emph{uC-Theorem} gives a general condition on the form of a constraint language to yield a $\PSPACE$-complete constrained synchronization problem.

 \medskip 
 
\item The \emph{Ideal Separation Theorem}. 
 In general, if the constraint language could be written as the union of two languages,
 and for one of them the constrained problem is hard, we cannot deduce hardness for the original languages.
 However, under certain circumstances, namely if the hard language is contained in a unique
 regular ideal language, we can infer hardness for the original languages.
 
\end{enumerate}

\noindent We apply these results to small constraint automata of up to three states.

 
\label{sec:preliminaries}

\section{General Notions and Definitions}

By $\Sigma$ we will always denote a \emph{finite alphabet}, i.e., 
a finite set of \emph{symbols}, or \emph{letters}.
A \emph{word} is an element of the free monoid $\Sigma^*$, i.e., the set of all finite sequences
with concatenation as operation. For $u, v \in \Sigma^*$, we will 
denote their\textbf{} concatenation by $u\cdot v$, but often we will omit the concatenation symbol
and simply write $uv$. The subsets of $\Sigma^*$
are also called \emph{languages}. By $\Sigma^+$ we denote the set of all words of non-zero length.
We write $\varepsilon$ for the empty word, and for $w \in \Sigma^*$ we denote by $|w|$
the length of $w$. Let $L \subseteq \Sigma^*$,
then $L^* = \bigcup_{n \ge 0} L^n$, with $L^0 = \{\varepsilon\}$
and $L^n = \{ u_1 \cdots u_n \mid u_1, \ldots, u_n \in L \}$ for $n > 0$,
denotes the \emph{Kleene star} of $L$.
For some language $L\subseteq \Sigma^*$,
we denote by $\pref(L) = \{ w \mid \exists u \in \Sigma^* : wu \in L \}$,
$\suff(L) = \{ w \mid \exists u \in \Sigma^* : uw \in L \}$
and $\factor(L) = \{ w \mid \exists u,v \in \Sigma^* : uwv \in L \}$
the set of \emph{prefixes}, \emph{suffixes} and \emph{factors} of words in $L$.
The language $L$ is called \emph{prefix-free} if
for each $w \in L$ we have $\pref(w) \cap L = \{w\}$.
If $u, w \in \Sigma^*$, a prefix $u \in \pref(w)$ is called a \emph{proper prefix} if
$u \ne w$.
%
%
A language $L \subseteq \Sigma^*$ is called a \emph{right} (\emph{left}-) \emph{ideal}
if $L = L\cdot \Sigma^*$ ($ = \Sigma^* \cdot L$), or
a \emph{two-sided ideal} (or simply an \emph{ideal} for short), if $L$ is both, a right and a left ideal.
A language $L \subseteq \Sigma^*$ is called \emph{bounded}, 
if there exist words $w_1, \ldots, w_n \in \Sigma^*$
such that $L \subseteq w_1^* \cdots w_n^*$.
%
%
%
%

Throughout the paper, we consider deterministic finite automata (DFAs).
Recall that a DFA~$\mathcal A$ is a tuple
$\mathcal A = (\Sigma, Q, \delta, q_0, F)$,
where the alphabet $\Sigma$ is a finite set of input symbols,~$Q$ is the finite state set, with start state $q_0 \in Q$, and final state set $F \subseteq Q$.
The transition function $\delta \colon Q\times \Sigma \to Q$ extends to words from $\Sigma^*$ in the usual way. The function $\delta$ can be further extended to sets of states in the following way. For every set $S \subseteq Q$ and $w \in \Sigma^*$, we set $\delta(S, w) := \{\,\delta(q, w) \mid q \in S\,\}$.
We sometimes refer to the function $\delta$ as a relation and we identify a transition $\delta(q, \sigma) = q'$ with the tuple $(q, \sigma, q')$. 
We call $\mathcal A$ \emph{complete} if~$\delta$ is defined for every $(q,a)\in Q \times \Sigma$; if $\delta$ is undefined for some $(q,a)$, the automaton~$\mathcal A$ is called \emph{partial}.
%
The set $L(\mathcal A) = \{\, w \in \Sigma^* \mid \delta(q_0, w) \in F\,\}$ denotes the language
\emph{recognized} by $\mathcal A$.

A \emph{semi-automaton} is a finite automaton without a specified start state
and with no specified set of final states.
The properties of being \emph{deterministic}, \emph{partial}, and \emph{complete} of semi-automata are defined as for DFA.
When the context is clear, we call both deterministic finite automata and semi-automata simply \emph{automata}.
We call a deterministic complete semi-automaton a DCSA
and a partial deterministic finite automaton a PDFA for short. If we want to add an explicit initial state $r$ and an explicit set of  final states $S$ to a DCSA $\mathcal A$, which changes it to a DFA, we use the notation $\mathcal A_{r,S}$.

A complete automaton $\mathcal A$ is called \emph{synchronizing} if there exists a word $w \in \Sigma^*$ with $|\delta(Q, w)| = 1$. In this case, we call $w$ a \emph{synchronizing word} for $\mathcal A$.
We call a state $q\in Q$ with $\delta(Q, w)=\{q\}$ for some $w\in \Sigma^*$ a \emph{synchronizing state}.

For an automaton $\mathcal A = (\Sigma, Q, \delta, q_0, F)$, 
we say that two states $q, q' \in Q$ are \emph{connected}, if one is reachable from the other,
i.e., we have a word $u \in \Sigma^*$ such that $\delta(q, u) = q'$. 
A subset $S \subseteq Q$ of states is called
\emph{strongly connected}, if all pairs from $S$ are connected.
A maximal strongly connected subset is called a \emph{strongly connected component}.
A state from which some final state is reachable is called \emph{co-accessible}.
An automaton $\mathcal A$ is called \emph{returning}, if for every state $q \in Q$, there exists a word $w \in \Sigma^*$ such that $\delta(q, w) = q_0$, where $q_0$ is the start state of $\mathcal A$.
A state $q \in Q$ such that for all $x \in \Sigma$ we have $\delta(q,x) = q$
is called a \emph{sink state}.

%
The set of synchronizing words forms a two-sided ideal. We will use this fact frequently
without further mentioning.

\begin{toappendix} 
The following obvious remark, stating that the set of synchronizing words
is a two-sided ideal, will be used frequently without further mentioning.

\begin{lemma}
	\label{lem:append_sync} 
	Let $\mathcal A = (\Sigma, Q, \delta)$ be a DCSA and $w\in \Sigma^*$ be a synchronizing word for $\mathcal A$. Then for every $u, v \in \Sigma^*$, the word $uwv$ is also synchronizing for~$\mathcal A$. 
\end{lemma}
\end{toappendix}

For a fixed PDFA $\mathcal B = (\Sigma, P, \mu, p_0, F)$,
we define the \emph{constrained synchronization problem}:

\begin{quote}
\begin{definition}
 \textsc{$L(\mathcal B)$-Constr-Sync}\\
\emph{Input}: DCSA $\mathcal A = (\Sigma, Q, \delta)$.\\
\emph{Question}: Is there a synchronizing word $w$ for $\mathcal A$ with  $w \in L(\mathcal B)$?
\end{definition}
\end{quote}

The automaton $\mathcal B$ will be called the \emph{constraint automaton}.
If an automaton $\mathcal A$ is a yes-instance of \textsc{$L(\mathcal B)$-Constr-Sync} we call $\mathcal A$ \emph{synchronizing with respect to $\mathcal{B}$}. Occasionally,
we do not specify $\mathcal{B}$ and rather talk about \textsc{$L$-Constr-Sync}.
We are going to inspect the complexity of this problem for different
(small) constraint automata.
The unrestricted synchronization problem, i.e., $\Sigma^*$\textsc{-Constr-Sync}
in our notation, is in $\PTIME$~\cite{Vol2008}.

We assume the reader to have some basic knowledge in computational complexity theory and formal language theory, as contained, e.g., in~\cite{HopMotUll2001}. For instance, we make use of  regular expressions to describe languages.
We also identify singleton sets with its elements. 
And we make use of complexity classes like $\PTIME$, $\NP$, or $\PSPACE$. 
With  $\leq^{\log}_m$ we denote a logspace many-one reduction.
If for two problems $L_1, L_2$ it holds that $L_1 \leq^{\log}_m L_2$ and $L_2 \leq^{\log}_m L_1$, then we write $L_1 \equiv^{\log}_m L_2$.

\section{Known Results on Constrained Synchronization}

Here we collect results from \cite{DBLP:conf/mfcs/FernauGHHVW19,DBLP:conf/cocoon/Hoffmann20,DBLP:conf/ictcs/Hoffmann20},
and some consequences, that will be used later.

\begin{lemma}[\cite{DBLP:conf/cocoon/Hoffmann20}]
	\label{lem:union}
	Let $\mathcal X$ denote any of the complexity classes
	 $\PSPACE$, $\NP$ and $\PTIME$.
	If $L(\mathcal B)$ is a finite union of languages $L(\mathcal B_1),
	L(\mathcal B_2), \dots, L(\mathcal B_n)$ such that for each $1\leq i\leq n$
 we have $L(\mathcal B_i)\textsc{-Constr-Sync}\in \mathcal X$, 
	then $L\textsc{-Constr-Sync}\in \mathcal X$.
\end{lemma}

The next result from \cite{DBLP:conf/mfcs/FernauGHHVW19}
states that the computational complexity is always in $\PSPACE$.

\begin{theorem}[\cite{DBLP:conf/mfcs/FernauGHHVW19}]
  \label{thm:L-contr-sync-PSPACE}
  For any constraint automaton $\mathcal B = (\Sigma, P, \mu, p_0, F)$
  the problem \textsc{$L(\mathcal B)$-Constr-Sync} is in $\PSPACE$.
\end{theorem}

\begin{toappendix} 
If $|L(\mathcal B)| = 1$ then $L(\mathcal B)\textsc{-Constr-Sync}$
is obviously in $\PTIME$. Simply feed this single word into the input
semi-automaton for every state and check if a unique state results.
Hence by Lemma \ref{lem:union} the next is implied.

\begin{lemma}
\label{lem:finite} 
 Let $\mathcal B = (\Sigma, P, \mu, p_0, F)$ be a constraint automaton
 such that $L(\mathcal B)$ is finite, then
 $L(\mathcal B)\textsc{-Constr-Sync} \in \PTIME$.
\end{lemma}
\end{toappendix}

In~\cite[Theorems 24, 25 and 26]{DBLP:conf/mfcs/FernauGHHVW19}, for a two-state partial constraint 
automaton with an at most ternary alphabet, the following complexity
classification was proven. In Section~\ref{sec:two-state}, we will extend
this result to arbitrary alphabets.

\begin{theorem}[\cite{DBLP:conf/mfcs/FernauGHHVW19}]
\label{thm:classification_MFCS_paper}
 Let $\mathcal B = (\Sigma, P, \mu, p_0, F)$
 be a PDFA.
 If $|P|\le 1$ or $|P| = 2$ and $|\Sigma|\le 2$, then $L(\mathcal B)\textsc{-Constr-Sync} \in \PTIME$.
 For $|P| = 2$ with $|\Sigma| = 3$, up to symmetry by renaming of the letters,
 $L(\mathcal B)\textsc{-Constr-Sync}$
 is $\PSPACE$-complete precisely in the following cases for $L(\mathcal B)$:
 $$
  \begin{array}{llll}
    a(b+c)^*        & (a+b+c)(a+b)^*  & (a+b)(a+c)^* & (a+b)^*c \\
    (a+b)^*ca^*     & (a+b)^*c(a+b)^* & (a+b)^*cc^*  & a^*b(a+c)^* \\
    a^*(b+c)(a+b)^* & a^*b(b+c)^*     & (a+b)^*c(b+c)^* & a^*(b+c)(b+c)^*
  \end{array}
 $$
 and polynomial time solvable in all other cases.
\end{theorem}

The next result from \cite[Theorem 17]{DBLP:conf/mfcs/FernauGHHVW19} will also be useful to single out certain polynomial time solvable
cases.

\begin{theorem}[\cite{DBLP:conf/mfcs/FernauGHHVW19}] 
\label{thm:returning_poly_alg}
	If  
    $\mathcal B$
	is returning, then  $\textsc{$L(\mathcal B)$-Constr-Sync}\in\PTIME$.	
\end{theorem}

The next result allows us to assume a standard form for two-state constraint automata.
We will prove an analogous result for three-state constraint automata
in Section~\ref{sec:three-states}.

\begin{lemma}[\cite{DBLP:conf/mfcs/FernauGHHVW19}]
	\label{lem:start-and-final}
	Let  $\mathcal B = (\Sigma, P, \mu)$ be a partial deterministic semi-automaton with two states, i.e., $P=\{1,2\}$. Then, for each $p_0\in P$ and each $F\subseteq P$, either  $L(\mathcal{B}_{p_0,F})\textsc{-Constr-Sync}\in\PTIME$, or  $L(\mathcal{B}_{p_0,F})\-\textsc{-Constr-Sync}\equiv^{\log}_m L(\mathcal{B}')\-\textsc{-Constr\--Sync}$ for  a PDFA $\mathcal{B}'=(\Sigma, P,\mu',1,\{2\})$.
\end{lemma}

\begin{toappendix}
The polycyclic languages, and polycyclic automata, were introduced in~\cite{DBLP:conf/ictcs/Hoffmann20}.

\begin{definition}[\cite{DBLP:conf/ictcs/Hoffmann20}]
\label{def:polycyclic_automata} 
 A PDFA $\mathcal{B} = (\Sigma, P, \mu, p_0, F)$
 is called \emph{polycyclic},
 if for any $p\in P$
 we have $L(\mathcal{B}_{p,\{p\}})\subseteq \{u_p\}^*$ for some $u_p \in \Sigma^*$.
 A language $L \subseteq \Sigma^*$
 is called \emph{polycyclic},
 if there exists a polycylic PDFA recognizing~it.
\end{definition}


In~\cite[Theorem 2]{DBLP:conf/ictcs/Hoffmann20}, it was shown that
for polycyclic languages, the constrained synchronization problem
is always in $\NP$. 

\begin{theorem}[\cite{DBLP:conf/ictcs/Hoffmann20}]
 If the PDFA $\mathcal B = (\Sigma, P, \mu, p_0, F)$ is polyclic,
 then we have $L(\mathcal B)\textsc{-Constr-Sync}\in \NP$.
\end{theorem}


 By this result, we can also deduce that for bounded constraint languages,
 the constrained synchronization problem is in $\NP$.
\end{toappendix}

 The next result combines results from~\cite{DBLP:conf/ictcs/Hoffmann20}
 and~\cite{GinsburgSpanier66} to show that for bounded constrained languages,
 the constrained synchronization problem is in $\NP$.

\begin{theoremrep}
\label{thm:bounded_in_NP}
 For bounded constraint languages, the constrained synchronization
 problem is in $\NP$.
\end{theoremrep}
\begin{proof}
By a result of Ginsburg~\cite[Theorem 1.2]{GinsburgSpanier66},
 every bounded regular language $L \subseteq w_1^* \cdots w_n^*$,
 with non-empty words $w_i, i \in \{1,\ldots,n\}$, 
 is a finite union 
 of languages of the form $X_1 \cdots X_n$
 with $X_i \subseteq w_i^*$ being regular.
 Define a homomorphism $\varphi : \{a\}^* \to \Sigma^*$
 by $\varphi(a) = w$. Then, for $i \in \{1,\ldots,n\}$,
 we have $\{ a^n \mid w^n \in X_i \} = \varphi^{-1}(X_i)$.
 Hence, this language is regular and from a unary automaton
 for it, we can easily construct a polycyclic automaton for $X_i$.
 By closure properties
of the polycyclic languages under concatenation and union~\cite[Proposition 5 and 6]{DBLP:conf/ictcs/Hoffmann20}, it is implied that $L$
is polycyclic.~\qed
\end{proof}

The following condition will be useful to single out, for bounded constraint
languages, those problems that are $\NP$-complete.

\begin{proposition}[\cite{DBLP:conf/ictcs/Hoffmann20}]
\label{prop:NPc}
 Suppose we find $u, v \in \Sigma^*$ 
 such that we can write
$
 L = u v^* U
$
 for some non-empty language $U \subseteq \Sigma^*$
 with 
 $$
  u \notin \factor(v^*), \quad
  v \notin \factor(U), \quad
  \pref(v^*) \cap U = \emptyset.
 $$
 Then $L\textsc{-Constr-Sync}$ is $\NP$-hard.
\end{proposition}

\section{General Results}
\label{sec:general_results}



Here, we state various general results, among them our three main theorems: the Ideal Separartion Theorem, the $UV^*W$-Theorem and the $uC$-Theorem.
The first result is a slight generalization of a Theorem
from~\cite[Theorem 27]{DBLP:conf/mfcs/FernauGHHVW19}.

\begin{theoremrep}
\label{thm:hom_lower_bound_complexity}
 Let $\varphi \colon \Sigma^* \to \Gamma^*$ be a homomorphism
 and $L \subseteq \Sigma^*$.
 Then $\varphi(L)\textsc{-Constr-Sync} \le_m^{\log} L\textsc{-Constr-Sync}$.
\end{theoremrep}
\begin{proof}
 
 

 Let $\mathcal A = (\Gamma, Q, \delta)$ be a DCSA.
	We want to know if it is synchronizing with respect to~$\varphi(L)$.
	Build the automaton $\mathcal A' = (\Sigma, Q, \delta')$ according to the rule
	$$
	\delta'(p, x) = q\quad\mbox{if and only if}\quad
        \delta(p, \varphi(x)) = q,$$ for $x\in\Sigma^*$.
	As~$\varphi$ is a mapping,~$\mathcal A'$ is indeed deterministic and
        complete, as~$\mathcal A$ is a DCSA. As the homomorphism~$\varphi$ is
        independent of $\mathcal A$, automaton~$\mathcal A'$ can be constructed from~$\mathcal A$
        in logarithmic space.  Next we prove that the 
        translation is indeed a reduction.

        If $u \in \varphi(L)$ is some synchronizing word for~$\mathcal A$, then there is
        some~$s\in Q$ such that $\delta(r,u)=s$, for all $r\in Q$.  By
        choice of $u$, we find $w \in L$ such that $u = \varphi(w)$. As with
        $\delta(r,\varphi(w))=s$, it follows $\delta'(r,w)=s$, hence~$w$
        is a synchronizing word for~$\mathcal A$.
        Conversely, if $w \in L$ is a synchronizing word
        for~$\mathcal A'$, then there is some $s\in Q$ such that
        $\delta'(r,w)=s$, for all $r\in Q$.  Further, $\varphi(w)$ is
        a synchronizing word for $\mathcal A$, as by definition for all $r\in
        Q$, we have $\delta(r,\varphi(w))=s$. \qed
\end{proof}


We will also need the next slight generalization of a Theorem
from~\cite[Theorem 14]{DBLP:conf/mfcs/FernauGHHVW19}.

\begin{theoremrep}
	\label{thm:add-stuff-general}
	Let $L, L' \subseteq \Sigma^*$.
	If $L \subseteq \factor(L')$
	and $L' \subseteq \factor(L)$, then
	 \[ L\textsc{-Constr-Sync}\equiv^{\log}_m L'\textsc{-Constr-Sync}. \]
\end{theoremrep}
\begin{proof}
 Let $\mathcal A = (\Sigma, Q, \delta)$ be a DCSA.
 If $\mathcal A$ has a synchronizing $u \in L$,
 then by assumption we find $x,y \in \Sigma^*$
 such that $xuy \in L'$, and by Lemma~\ref{lem:append_sync}
 the word $xuy$ also synchronizes $\mathcal A$.
 Similarly, if $\mathcal A'$ has a synchronizing word in $L'$,
 then we know it has one in $L$.\qed
\end{proof}

\begin{toappendix} 
\begin{remark}\label{remark:add-stuff} \todo{genauer?}
Considering a PDFA~$\mathcal{B} = (\Sigma, P, \mu, p_0, F)$, we conclude:
If $p \in P$ is co-accessible, then $L(\mathcal B)\textsc{\--Constr\--Sync}\equiv^{\log}_m L(\mathcal{B}_{p_0, F\cup \{p\}})\textsc{\--Constr\--Sync}$.
\end{remark}
\end{toappendix}

Next, we state a result on how we can combine languages using concatenation,
while still getting polynomial time solvable problems. 
Another result, namely Theorem~\ref{thm:uC_PSPACE-hard},
is contrary in the sense that it states conditions for which the concatenation
yields $\PSPACE$-hard problems.

\begin{toappendix}
 
 In the proof of Theorem~\ref{thm:UVW-theorem}
 we need the next observation, which is simple 
 to see, as every accepting path has to use only co-accessible
 states.

\begin{lemma}
\label{lem:co-accessible-states-only} 
 Let $\mathcal B = (\Sigma, P, \mu, p_0, F)$ be a PDFA. If $p \in P \setminus \{p_0\}$ is not co-accessible,
 then let $\mathcal B' = (\Sigma, P', \mu', p_0, F')$
 be the automaton constructed from $\mathcal B$ by removing the state $p$, i.e.,
 (i) $P' = P \setminus \{p\}$, 
 (ii) $\mu' = \mu \cap ( P' \times \Sigma \times P' )$, and
 (iii) $F' = F \setminus \{p\}$.
 Then
 $L(\mathcal B) = L(\mathcal B')$.
\end{lemma}
\end{toappendix}

\begin{theoremrep}[$UV^*W$-Theorem]
\label{thm:UVW-theorem}
 Let $U, V, W \subseteq \Sigma^*$ be regular and $\mathcal B = (\Sigma, P, \mu, p_0, \{p_0\})$
 be a PDFA, whose initial state equals its single final state, such that
 \begin{enumerate}
 \item 
   $V = L(\mathcal B)$, 
 \item $U \subseteq \suff(V)$ and
 \item $W \subseteq \pref(V)$.
 \end{enumerate}
 Then $(UVW)\textsc{-Constr-Sync}\in \PTIME$.

\end{theoremrep}
\begin{proof}
 Let $\mathcal A = (\Sigma, Q, \delta)$ be a DCSA.

 \begin{enumerate}
 \item If $v \in V$ synchronizes $\mathcal A$, then, for $u \in U$ and $w \in W$,
  also $uvw \in UVW$ synchronizes $\mathcal A$ by Lemma~\ref{lem:append_sync}.
  
 \item Conversely, suppose $uvw \in UVW$, with $u \in U, v \in V, w \in W$,
  synchronizes $\mathcal A$.
  By assumption, there exist $x,y \in \Sigma^*$
  such that $xu \in V$ and $wy \in V$.
  As $\mu(p_0, xu) = p_0$, $\mu(p_0, wy) = p_0$
  and $\mu(p_0, v) = p_0$.
  So, $\mu(p_0, xuvwy) = p_0$, i.e., $xuvwy \in V$
  and by Lemma~\ref{lem:append_sync}
  the word $xuvy \in V$ also synchronizes $\mathcal A$.
 \end{enumerate}
 
 Hence, $\mathcal A$ has a synchronizing word in $V$
 if and only if it has a synchronizing word in $UVW$.
 So, deciding synchronizability with respect to $UVW$
 is at most as hard as deciding synchronizability
 with respect to $V$.
 By Lemma~\ref{lem:co-accessible-states-only}, we can assume
 $\mathcal B$ has only co-accessible states.
 But as $p_0$ is the only final state, this in particular
 implies that $\mathcal B$ is returning.
 Hence, by Theorem~\ref{thm:returning_poly_alg},
 $V\textsc{-Constr-Sync}\in\PTIME$ and by the above
 considerations $(UVW)\textsc{-Constr-Sync}\in\PTIME$. \qed
\end{proof}

\begin{toappendix}
  
  In the proof of the Ideal Separation Theorem, we need the following
  observation. Basically, it combines the two observations
  that for an ideal language, we have a single final sink state
  in the minimal automaton, and that every state is reachable from 
  the start state and combining an accepted word with some word reaching this state
  gives an accepted word.

  \begin{lemma}[\cite{DBLP:journals/corr/Maslennikova14,DBLP:conf/cwords/GusevMP13}]
  \label{lem:ideal_language} 
   Let $L \subseteq \Sigma^*$ be an ideal language. Then $L$
   is precisely the set of synchronizing words for the minimal automaton
   of $L$.
  \end{lemma}
  
\end{toappendix}

\begin{remark}
 Note that in Theorem~\ref{thm:UVW-theorem}, $U = \{\varepsilon\}$
 or $W = \{\varepsilon\}$ is possible. In particular, 
 $L(\mathcal B)\textsc{-Constr-Sync}\in \PTIME$
 for every PDFA $\mathcal B = (\Sigma, P, \mu, p_0, \{p_0\})$.
\end{remark}

The next theorem is useful, as it allows us to show $\PSPACE$-hardness
by reducing the problem, especially ones
that are written as unions, to known $\PSPACE$-hard problems.
Please see Example~\ref{ex:gen_results},
or the proof sketch of Theorem~\ref{thm:three_states}, for applications.

\begin{theoremrep}[Ideal Separation Theorem]
\label{thm:ideal_hardness}
 Let $I \subseteq \Sigma^*$ be a fixed regular ideal language.
 Suppose $L \subseteq \Sigma^*$ is any regular language,
 then
 $$
  (I \cap L)\textsc{-Constr-Sync} \le_{m}^{\log} L\textsc{-Constr-Sync}.
 $$
 In particular, let $u \in \Sigma^*$ and $L \subseteq \Sigma^*$. 
 Then $(L\cap \Sigma^*u\Sigma^*)\textsc{-Constr-Sync} \le_{m}^{\log} L\textsc{-Constr-Sync}$.
\end{theoremrep}
\begin{proof}
 By Lemma~\ref{lem:ideal_language},
 we can suppose we have a DCSA $\mathcal C = (\Sigma, T, \eta)$ 
 whose set of synchronizing words is precisely $I$.
 Let $\mathcal A = (\Sigma, Q, \delta)$ be an input DCSA
 for which we want to know if it has a synchronizing word
 in $L \cap I$. Set $\mathcal A' = (\Sigma, Q \times T, \gamma)$
 with
 $
  \gamma( (q, t), x) = (\delta(q,x), \eta(t,x))
 $
 for $x \in \Sigma$.
 We claim that $\mathcal A'$ has a synchronizing word in $L$
 if and only if $\mathcal A$ has a synchronizing word in $L \cap I$.
 
 \begin{enumerate}
 \item Suppose we have some $w \in L$ such that $|\gamma(Q\times T, w)| = 1$.
 Then, by construction of $\mathcal A'$,
 $\delta(Q, w) = q$ for some $q \in Q$ and $\eta(T, w) = s$ for some $s \in T$.
 The last equation implies $w \in I$ and the former that $w$ is a synchronizing word for~$\mathcal A$.
 
 \item Suppose we have some $w \in I \cap L$ and $q \in Q$ such that $\delta(Q, w) = \{q\}$.
  Then, as $w \in I$, we have $\eta(T, w) = s$ for some $s \in T$. So,
  $\gamma(Q\times T, w) = \{ (\delta(q, w), \eta(t, w)) \mid q \in Q, t \in T \} = \{ (q, t) \}$, i.e. $w$ synchronizes $\mathcal A'$.
 \end{enumerate}
 Hence, we can solve $(I \cap L)\textsc{-Constr-Sync}$
 with $L\textsc{-Constr-Sync}$, where the reduction, as $I$, and hence its minimal automaton,  is fixed,
 could be done in polynomial time, as it is essentially the product automaton
 construction of the input DCSA with the minimal automaton of $I$.
 
 \medskip 
 
 For the additional sentence, note that, for $u \in \Sigma^*$,
 the language $\Sigma^*u \Sigma^*$ is an ideal.~\qed
\end{proof}

Most of the time, we will apply Theorem~\ref{thm:ideal_hardness}
with \emph{principal ideals}, i.e., ideals of the form $\Sigma^* u \Sigma^*$
for $u \in \Sigma^*$.
The next proposition is a broad generalization of arguments
previously used to establish $\PSPACE$-hardness~\cite{DBLP:conf/mfcs/FernauGHHVW19,DBLP:conf/cocoon/Hoffmann20}.

\begin{theoremrep}[uC-Theorem] \todo{Gamma in zweiten Teil noch, umbenennen uCu?}
\label{thm:uC_PSPACE-hard}
%
 Suppose $u \in \Sigma^+$ is a non-empty word.
 \begin{enumerate}
 \item  Let $C \subseteq \Sigma^*$ be a \emph{finite} prefix-free set of cardinality at least two
 with $C^* \cap \Sigma^* u \Sigma^* = \emptyset$.
 
 \item Let $\Gamma\subseteq \Sigma$ be such that $u$ uses
 at least one symbol not in $\Gamma$. More precisely, if $u = u_1 \cdots u_n$ with $u_1,\ldots, u_n \in \Sigma$,
 then $\{ u_1, \ldots, u_n \}\setminus \Gamma\ne \emptyset$.
 \end{enumerate}
 %
 %
 Then, the problem $(\Gamma^*uC^*)\textsc{-Constr-Sync}$
 is $\PSPACE$-hard.
 If, additionally, we have $\suff(u) \cap \pref(u) = \{\varepsilon,u\}$
 and the following is true:
 \begin{quote}
  There exists $x \in C$ such that,  
   for $v,w \in \Sigma^*$, if $vxw \in (C \cup \{u\})^*$, then $vx \in (C \cup \{u\})^*$.
 \end{quote}
 Then, $(C^*u\Gamma^*)\textsc{-Constr-Sync}$ is $\PSPACE$-hard.
\end{theoremrep}
\begin{proof} 

 We show both $\PSPACE$-hardness statements
 separately, where for the $\PSPACE$-hardness
 of $(C^*u\Gamma^*)$\textsc{-Constr}\textsc{-Sync}, we first
 show that $(C^*u^*)$\textsc{-Con\-str-Sync}
 is $\PSPACE$-hard and use this result
 to show it for the original constraint language $C^*u\Gamma^*$.
 
 \medskip 

 \noindent\underline{The problem $(\Gamma^*uC^*)\textsc{-Constr-Sync}$ is $\PSPACE$-hard.}
 
 \medskip 
 
  
 %
 Choose two distinct $x,y \in C$.
 Let $\Delta = \{a,b,c\}$ with $\Delta \cap \Sigma = \emptyset$
 and $\mathcal A = (\Delta, Q, \delta)$
 be a semi-automaton for which we want to know if it has a synchronizing word in $a(b+c)^*$.
 First, let $\mathcal C = (\Sigma, T, \eta, t_0, F)$ be the minimal automaton
 for $\Sigma^* u \Sigma^*$.
 As it only has to detect if $u$ is read at least once,
 we can deduce that $\mathcal C$ has a single final state, which
 is also a sink state. Write $F = \{t_f\}$.
 Set $S = \delta(Q, a)$ and fix an arbitrary $s' \in S$.
 Then construct $\mathcal A' = (\Sigma, Q', \delta')$
 with
 \[
  Q' = Q'' \cup S \times (T \setminus \{t_f\})
 \]
 for $Q'' = \{ q_w \mid w \in \pref(C) \setminus C, q \in Q \}$.
 We identify $Q$ with $\{ q_{\varepsilon} \mid q \in Q \}$, hence $Q \subseteq Q'$,
 but the states $q_w$, $w \ne \varepsilon$, are new and disjoint to\footnote{Note that by choice of
 notation, a correspondence between the states $q$ and $q_w$
 for $q \in Q$ and $w \in \pref(C) \setminus\{\varepsilon\}$ was set up, which
 will be used in the following constructions.} $Q$.
 For $q_w$ with $w \in \pref(C) \setminus C$ and $z \in \Sigma$ set
 \[
  \delta'(q_w, z) = \left\{
  \begin{array}{ll}
  \delta(q, b) & \mbox{if } wz = x; \\ 
  \delta(q, c) & \mbox{if } wz = y; \\
  q            & \mbox{if } wz \in C \setminus\{x,y\}; \\
  s'           & \mbox{if } wz \notin \pref(C); \\
  q_{wz}       & \mbox{if } wz \in \pref(C) \setminus C.
  \end{array}
  \right.
 \]
 and, for $q \in S$, $r \in T \setminus\{t_f\}$ and $z \in \Sigma$,\todo{intuition, skizze?}
 \[ 
  \delta'((q, r), z) = \left\{ 
   \begin{array}{ll}
     (q, \eta(r, z)) & \mbox{if } \eta(r,z) \ne t_f;  \\ 
     q_{\varepsilon} & \mbox{if } \eta(r,z) = t_f. \\ 
   \end{array}
   \right.
 \]
 Let $\varphi : \{b,c\}^* \to \{x,y\}^*$ with $\varphi(b) = x$ and $\varphi(c) = y$.
 Then we have for each $q \in Q$ and $w \in \{b,c\}^*$,
 by construction of $\mathcal A'$,
 \begin{equation}\label{eqn:Abar} 
     \delta'(q, \varphi(w)) = \delta(q, w).
 \end{equation}

 Finally, we show that we have a synchronizing word for $\mathcal A'$
 in $\Gamma^*uC^*$ if and only if $\mathcal A$
 has a synchronizing word in $a(b+c)^*$.
 
 \begin{enumerate}
 \item Suppose $\mathcal A'$ has a synchronizing word in $\Gamma^* u C^*$.
     
     


     Let the synchronizing word be $v_1uv_2$ with $v_1 \in \Gamma^*$
     and $v_2 \in C^*$.
     By construction of $\mathcal A'$,
     we have $\delta'(Q'', v_1) \in Q''$.
     Furthermore,
     as $C^* \cap \Sigma^* u \Sigma^*$,
     we have $\delta(Q'', u) = \{s'\}$,
     for if for some $w,w' \in \pref(C) \setminus C$
     we have $\delta'(q_w, v) = q_{w'}$,
     then $v$ is a factor of some word in $C^+$ by construction of $\mathcal A'$.
     As $u$ contains some symbol not in $\Gamma$,
     we have $\Sigma^* u \Sigma^* \cap \Gamma^* = \emptyset$,
     and so $\eta(t_0, z) \ne t_f$.
     Hence, $\delta'(S \times \{t_0\}, v_1) = S \times ( T \setminus \{ t_f \} )$.
     Also $\delta'(S \times T \setminus\{t_0, t_f\}, v_1) \subseteq S \times T \setminus\{t_f\} \cup Q''$.
     So,
     $
      \delta'(S \times \{t_0\}, v_1 u) = S.
     $
     Hence
     \begin{align*}
      \delta'(Q'' \cup S \times T\setminus\{t_f\}, v_1u) 
       & = \delta'(Q'', v_1u) \cup \delta'(S \times T\setminus\{t_f\}, v_1u) \\
       & = \{s'\} \cup S \\
       & = S.
     \end{align*}
     We can assume $v_2 \in \{x,y\}^*$, as
     for any $z \in C \setminus \{x,y\}$ and $q \in Q$, we
     have 
     $
      \delta'(q, z) = q.
     $
     Hence, if $v_2 = u_1 \cdots u_n$ with $u_1, \ldots, u_n \in C$,
     where this decomposition is unique as $C$ is prefix-free,
     we can remove all factors $u_i \in C\setminus\{x,y\}$, $i \in \{1,\ldots,n\}$,
     to get a new word $v_2' \in \{x,y\}^*$ such that
     $\delta'(Q'', v_2) = \delta(Q'', v_2')$.
     Let $w' \in \{b,c\}^*$ be such that $\varphi(w') = v_2$.
     As $S \subseteq Q$, by Equation~\eqref{eqn:Abar},
     and as $|\delta'(S, v_2)| = 1$,
     we find $|\delta(S, w')| = 1$.
     Hence $|\delta(Q, aw')| = 1$.
     
  \item Suppose $\mathcal A$ has a synchronizing word $w \in a(b+c)^*$.
  
   Write $w = av$ with $v \in (b+c)^*$.
   By construction of $\mathcal A'$,
   \[
    \delta'(Q'' \cup S \times (T \setminus \{t_f\}), u) = S
   \]
   by similar arguments as in the first case above.
   As, by assumption $\delta(Q, av) = \delta(S, v)$
   is a singleton set. Hence, by Equation~\eqref{eqn:Abar},
   we find that $\delta'(S, \varphi(v))$
   is a singleton set. Combining all arguments,
   we find that $u\varphi(v) \in \Gamma^*uC^*$
   is a synchronizing word for $\mathcal A$.
 \end{enumerate}
 So, we have reduced the problem of synchronization with the constraint language $a(b+c)^*$, which is
 known to be \PSPACE-complete by Theorem~\ref{thm:classification_MFCS_paper}, to our
 problem, which gives the claim. \qed

 \medskip 
 
 \noindent\underline{The problem $(C^*u\Gamma^*)\textsc{-Constr-Sync}$ is $\PSPACE$-hard.}
 
 \medskip 
 We show it first for the constraint $C^*u$, and then use this result
 to show it for $C^*u\Gamma^*$.

 Choose two distinct $x,y \in C$.
 Let $\Delta = \{a,b,c\}$ with $\Delta \cap \Sigma = \emptyset$
 and $\mathcal A = (\Delta, Q, \delta)$ be a semi-automaton
 for which we want to know if it has a synchronizing word
 in $(a+b)^*c$. As shown in Theorem~\ref{thm:classification_MFCS_paper},
 this is a $\PSPACE$-complete problem, and we will reduce
 it to our problem at hand. 
 Looking at the reduction
 used in~\cite{DBLP:conf/mfcs/FernauGHHVW19}, we see that 
 $((a+b)^*c)\textsc{-Constr-Sync}$ is $\PSPACE$-hard
 even for input semi-automata with a sink state $s \in Q$ that is only reachable
 by the letter $c$, i.e, if $q \in Q \setminus\{s\}$, $x \in \Delta$
 and $\delta(q, x) = s$ implies $x = c$.
 We will use this observation, where $s \in Q$.
 denotes the sink state of $\mathcal A$.

 Define $\mathcal A' = (\Sigma, Q', \delta')$
 with (note that $C \cup \{u\}$ is prefix-free)
 \[
  Q' = \{ q_w \mid w \in \pref(C \cup \{u\}) \setminus (C \cup \{u\}) \mbox{ and } q \in Q \setminus\{s\} \} \cup\{ s \}.
 \]
 We identify $Q \setminus\{s\}$
 with $\{ q_{\varepsilon} \mid q \in Q \setminus\{s\} \}$. Hence $Q \subseteq Q'$,
 but the states $q_w$, $w \ne \varepsilon$, are new and disjoint to\footnote{Note that by choice of
 notation, a correspondence between the states $q$ and $q_w$
 for $q \in Q$ and $w \in \pref(C) \setminus\{\varepsilon\}$ was set up, which
 will be used in the following constructions.} $Q$.
 For $q_w$, $q \in Q \setminus\{s\}$, with $w \in \pref(C \cup \{u\}) \setminus (C \cup \{u\})$ and $z \in \Sigma$ set
 \[
 %
  \delta'(q_w, z) = \left\{
  \begin{array}{ll}
  \delta(q, a) & \mbox{if } wz = x; \\ 
  \delta(q, b) & \mbox{if } wz = y; \\
  \delta(q, c) & \mbox{if } wz = u; \\ 
  q            & \mbox{if } wz \in C \setminus\{x,y\}; \\
  s            & \mbox{if } wz \notin \pref(C); \\
  q_{wz}       & \mbox{if } wz \in \pref(C) \setminus C; \\ 
  q_{wz}       & \mbox{if } wz \in \pref(u) \setminus \{u\}. 
  \end{array}
  \right.
 \]
 and $\delta'(s, z) = s$. Let
 $\varphi : \{a,b,c\}^* \to \Sigma^*$ 
 be the homomorphism given by $\varphi(a) = x$, $\varphi(b) = y$
 and $\varphi(c) = u$. Then, by construction of $\mathcal A'$,
 for any $q \in Q$ and $v \in \{a,b,c\}^*$, we have
 \begin{equation}\label{eqn:Cu_hom} 
  \delta(q, v) = \delta'(q, \varphi(v)).
 \end{equation}
 
 \begin{claiminproof}
  The semi-automaton $\mathcal A'$ has a synchronizing word in $C^*u$
  if and only if $\mathcal A$ has a synchronizing word in $(a+b)^*c$.
 \end{claiminproof}
 \begin{claimproof}
 First, suppose $w \in (a+b)^*c$ synchronizes $\mathcal A$.
 As $s$ is a sink state, we have $\delta(Q, w) = \{s\}$.
 With the assumptions $\suff(u) \cap \pref(u) = \{u,\varepsilon\}$
 and $C^* \cap \Sigma^* u \Sigma^* = \emptyset$, we can deduce
 that $\delta'(Q' \setminus Q, u) = \{s\}$.

   \begin{claiminproof}
   We must have $\delta'(Q' \setminus Q, u) = \{s\}$.
  \end{claiminproof}
  \begin{claimproof}
   First, assume $\delta'(q_{w_1}, u) = q'_{w_2}$ with $w_1, w_2 \in \pref(C \cup \{u\}) \setminus (\{ u, \varepsilon \} \cup C)$
   and $q, q' \in Q \setminus \{s\}$.
   In this case, we can write $w_1 u = w_3 w_2$
   with $w_3 \in (C \cup \{u\})^*$.
   Write $w_3 = w_4 w_5$ with $w_4 \in C \cup \{u\}$
   and $w_5 \in (C \cup \{u\})^*$.
   As $w_1 \in \pref(C\cup \{u\}) \setminus (C\cup\{u\})$
   and $C \cup \{u\}$ is a prefix code\footnote{A prefix code is the same as a prefix-free set. I only call them codes here to emphasize that every word in $(C\cup\{u\})^+$ has a unique factorization,
   the defining property of codes.}, $w_4$ could not be a prefix of $w_1$
   and so $w_1$ must be a proper prefix of $w_4$.
   This yields $w_4 \ne u$, for otherwise some non-trivial 
   suffix of $u$ is a non-trivial prefix of $u$, which is excluded.
   We have $|w_1| + |u| = |w_4| + |w_5| + |w_2|$,
   and as $w_1$ is a proper prefix of $w_4$,
   this yields $|u| > |w_5| + |w_2|$.
   Then, as $|w_5| < |u|$, we must have $w_5 \in C^*$.
   So, as $w_4 \in C$ and $w_5 \in C^*$,
   we have $w_3 = w_4 w_5 \in C^*$.
   Now, as $w_2$ is a non-trivial suffix of $u$, we
   also must have $w_2 \in \pref(C)\setminus C$.
   However, $w_3 \in C^*$ and $w_2 \in \pref(C) \setminus C$
   would imply $\Sigma^*u\Sigma^* \cap C^* \ne \emptyset$,
   which is excluded. Hence, this case is not possible.

   Next, assume $\delta'(q_{w_1}, u) = q$, $q \in Q \setminus \{s\}$ and $w_1 \in \pref(C \cup \{u\}) \setminus (\{ u, \varepsilon \} \cup C)$.
   In that case $w_1u \in (C \cup \{u\})^*$. 
  Write $w_1u = w_2 w_3$ with $w_2 \in C \cup \{u\}$ and $w_3 \in (C\cup \{u\})^*$.
  Then, as $C \cup \{u\}$ is a prefix code, $w_1$ is a proper prefix of $w_2$.
  Also, a non-trivial suffix of $w_2$ is a non-trivial prefix of $u$,
  which implies $w_2 \ne u$.
  As $|w_3| < |u|$, we also find $w_3 \in C^*$.
  But then, as $\Sigma^* u \Sigma^* \cap C^* = \emptyset$,
  this is not possible and so the above could not happen.
  
  Hence, the only case left is that for any $q_{w_1}$
  with $w_1 \in \pref(C \cup \{u\}) \setminus (C \cup \{u\})$
  we must have $\delta'(q_{w_1}, u) = s$, which shows the claim.
  \end{claimproof}

  \medskip

  \todo{genauer?}
  For $x \in C$, which we chose above in the construction,
  assume it is the word in $C$ such that, for $v,v' \in \Sigma^*$, if $vxv' \in (C \cup \{u\})^*$,
  then $vx \in (C \cup \{u\})^*$.
  Let $q_v \in Q' \setminus Q$
  with $v \in \pref(C \cup \{u\}) \setminus (C \cup \{u\})$.
  By the operational mode of $\mathcal A'$,
  if $\delta(q_v, x) \notin Q$, then there exists a non-empty $v \in \Sigma^+$
  such that $vxv' \in (C \cup \{u\})^*$. By assumption,
  this yields $vx \in (C \cup \{u\})^*$.
  However, also by the operational mode of $\mathcal A'$,
  this yields $\delta(q_v, x) \in Q$.
  So, we must have $\delta(q_v, x) \in Q$.
  Then, together with Equation~\eqref{eqn:Cu_hom}, we find
  \[
  \delta'(Q', x\varphi(w)) \subseteq \delta(Q,w) = \{s\},
 \]
 and the word $x\varphi(w) \in C^*u$ synchronizes $\mathcal A'$.

  Now, conversely, suppose $w = vu$ with $v \in C^*$ synchronizes $\mathcal A'$.
  Then, as $s$ is a sink state, $\delta'(Q', w) = \{s\}$.
  Let $v' \in \{x,y\}^*$ be the word that results out of $v$
  by deleting all factors that are words in $C \setminus \{x,y\}$.
  Note that, as $C$ is prefix-free, we have a unique factorization
  for $v$.
  As any word in $C \setminus \{x,y\}$ maps every state in $Q$
  to itself, we have, for any $q \in Q$,
  that $\delta'(q, v) = \delta(q, v')$.

  Let $q \in Q \setminus \{s\}$.
  Let $\varphi : \{a,b,c\}^* \to \Sigma^*$
  with $\varphi(a) = x$, $\varphi(b) = y$
  and $\varphi(c) = u$ and let $w' \in \{a,b,c\}^*$
  with $\varphi(w') = v'u$.
  By Equation~\eqref{eqn:Cu_hom},
  \[
   \delta'(q, v'u) = \delta(q, w').
  \]
  Hence, $\{s\} = \delta(Q, v'u) = \delta(Q, w')$
  and so $w'$ synchronizes $\mathcal A$.
 \end{claimproof}

 However, up to now we have only shown hardness
 for the constraint $C^*u$, but we want it for $C^*u\Gamma$
 as stated in the theorem.
 
 \begin{claiminproof}
  The automaton $\mathcal A'$
  has a synchronizing word in $C^*u$ if and only if it has a synchronizing word
  in $C^*u\Gamma^*$.
 \end{claiminproof}
 \begin{claimproof}
  As $C^*u \subseteq C^*u\Gamma^*$, we only have to show one implication.
  So, assume $\mathcal A'$
  has a synchronizing word $w \in C^*u \Gamma^*$.
  Write $w = x'u y'$ with $x' \in C^*$ and $y' \in \Gamma^*$.
  By the assumptions, we have $u \notin \Gamma^*$.
  If $q \in Q \setminus \{s\}$, as, by assumption about $\mathcal A$,
  only $u$ maps any such state
  to $s$, we have $\delta'(q, y') \ne s$.
  As shown above, $\delta'(Q', x'u) \subseteq Q$.
  Combining both observations, we must have $\delta'(Q', x'u) = \{s\}$,
  for otherwise, we cannot have $\delta'(Q', x'uy') = \{s\}$.
 \end{claimproof}
 
 \noindent This finishes the proof.~\qed 
\end{proof}

\begin{example}
\label{ex:abcbc_MFCS19}
 Set $L = (a+b)^*ac(b+c)^*$.
 Using Theorem~\ref{thm:uC_PSPACE-hard}
 with $\Gamma = \{a,b\}$, $u = ac$
 and $C = \{b,c\}$ gives $\PSPACE$-hardness.
 Hence, by Theorem~\ref{thm:L-contr-sync-PSPACE},
 it is $\PSPACE$-complete.
 Note that $(a+b)^*c(b+c)^* \cap \Sigma^*ac\Sigma^* = L$.
 Hence, together with Theorem~\ref{thm:ideal_hardness},
 we get $\PSPACE$-completeness for $(a+b)^*c(b+c)^*$.
 For the latter language, this was already shown in~\cite{DBLP:conf/mfcs/FernauGHHVW19}, as stated in Theorem~\ref{thm:classification_MFCS_paper},
 by more elementary means, i.e., by giving a reduction from a different problem.
\end{example}

\begin{example}
\label{ex:gen_results}
 For the following $L \subseteq \{a,b\}^*$
 we have that $L\textsc{-Constr-Sync}$
 is \PSPACE-hard. For the first two, this is implied by a straightforward application
 of Theorem~\ref{thm:uC_PSPACE-hard}, for the last one a more detailed proof is given.
 \begin{enumerate}
 \item $L= \Gamma^*aa(ba+bb)^*$ for $\Gamma\subseteq\{b\}$. 
 \item $L= \Gamma^*aba(a+bb)^*$ for $\Gamma\subseteq\{b\}$.
 \item $L= b^*a(a+ba)^*$.
   Then $L = b^*bba(a+ba)^* \cup ba(a+ba)^* \cup a(a+ba)^*$.
   Set $U = L \cap \Sigma^*bba\Sigma^* = b^*bba(a+ba)^*$.
   By Theorem~\ref{thm:ideal_hardness}, 
   \[ 
   U\textsc{-Constr-Sync} \le_m^{\log} L\textsc{-Constr-Sync}.
   \]
   For $U$, with $\Gamma = \{b\}$, $u = bba$ and $C = \{a,ba\}$
   and Theorem~\ref{thm:uC_PSPACE-hard},
   we find that $U\textsc{-Constr-Sync}$ is $\PSPACE$-hard.
   So, $L\textsc{-Constr-Sync}$ is also $\PSPACE$-hard
  \end{enumerate}
\end{example}

\section{Application to Small Constraint Automata}

%
%
%
%
%
%
%
%
%
%
%
%
%
%

Here, we apply the results obtained in Section~\ref{sec:general_results}.
In Subsection~\ref{sec:two-state} we will give a complete overview of the complexity
landscape for two-state constraint automata over an arbitrary alphabet, thus
extending a result from~\cite{DBLP:conf/mfcs/FernauGHHVW19}, where it was only proven
for an at most ternary alphabet. In Subsection~\ref{sec:three-states} we will
give a complete overview of the complexity landscape for three-state 
constraint automata over a binary alphabet, the least number of states over a binary
alphabet such that we get $\PSPACE$-complete and $\NP$-complete constrained
synchronization problems~\cite{DBLP:conf/mfcs/FernauGHHVW19}.

\medskip

\noindent\underline{Notational Conventions in this Section:} 
Let $\mathcal B = (\Sigma, P, \mu, p_0, F)$ be a constraint PDFA
with $|P| = n$. Here, we will denote the states by natural numbers
$P = \{1,\ldots, n\}$, and we will assume that $1$ always denotes the
start state, i.e., $p_0 = 1$. In this section, $\mathcal B$
will always denote the fixed constraint PDFA. 
By Lemma~\ref{lem:start-and-final},
for $|P| = 2$, we can assume $F = \{2\}$. We will show in Section~\ref{sec:three-states},
stated in Lemma~\ref{lem:start-final-three-states}, that also for $|P| = 3$
we can assume $F = \{3\}$. So, if nothing else is said, by default we will
assume $F = \{n\}$ in the rest of this paper.
Also, for a fixed constraint automaton\footnote{Note that this notation
only makes sense with respect to a fixed alphabet and a fixed automaton, or said
differently we have implicitly defined a function dependent on both of these parameters. But every more formal way of writing this might be cumbersome, and as the automaton
used in this notation is always the (fixed) constraint automaton, in the following, usage of this notation should pose no problems. It is just a shorthand whose usage is restricted
to the next two sections.}, 
we set $\Sigma_{ij}:=\{\,a\in\Sigma\mid \mu(i,a)=j\,\}$ for $1\le i,j\le n$.
As $\mathcal B$ is deterministic, $\Sigma_{i1}\cap\Sigma_{i2}=\emptyset$.

\subsection{Two States and Arbitrary Alphabet}
\label{sec:two-state}


Let $\mathcal B = (\Sigma, P, \mu, p_0, F)$
be a two-state constraint PDFA. 
Recall the definitions
of the sets $\Sigma_{i,j}$, $1 \le i,j \le 2$ and that here, by our notational conventions,
$P = \{1,2\}$, $p_0 = 1$ and $F = \{2\}$. 
In general, for two states, we have
\[
 L(\mathcal B) = (\Sigma_{1,1}^* \Sigma_{1,2} \Sigma_{2,2}^*  \Sigma_{2,1})^*\Sigma_{1,1}^* \Sigma_{1,2} \Sigma_{2,2}^*.
\]
First, as shown in~\cite{DBLP:conf/mfcs/FernauGHHVW19},
for two-state constraint automata, some easy cases could be excluded from further analysis 
by the next result, as they give polynomial time solvable instances.

\begin{proposition}[\cite{DBLP:conf/mfcs/FernauGHHVW19}]
\label{prop:easy_cases}
	If one of the following conditions hold, 
	then  $L(\mathcal{B}_{1,\{2\}})\-\textsc{-Constr\--Sync}\in\PTIME$:
	(1) $\Sigma_{1,2}=\emptyset$, (2) $\Sigma_{2,1}\neq\emptyset$, (3) $\Sigma_{1,1}\cup\Sigma_{1,2}\subseteq \Sigma_{2,2}$, or (4) $\Sigma_{1,1}\cup\Sigma_{2,2}=\emptyset$. 
	\end{proposition}
	
Next, we will single out those cases that give $\PSPACE$-hard problem 
in Lemma~\ref{lem:Sigma11_2_PSPACE-hard}
and Lemma~\ref{lem:Sigma22_2_PSPACE-hard}.
Finally, in Theorem~\ref{thm:two-state-classification}
we will combine these results and show that the remaining cases all give polynomial
time solvable instances.

\begin{lemma}
\label{lem:Sigma22_2_PSPACE-hard}
 Suppose $(\Sigma_{1,1} \cup \Sigma_{1,2}) \setminus \Sigma_{2,2} \ne \emptyset$,
 $\Sigma_{1,2} \ne \emptyset$, $\Sigma_{2,1} = \emptyset$
 and $|\Sigma_{2,2}|\ge 2$.
 Then $L(\mathcal B)\textsc{-Constr-Sync}$ is $\PSPACE$-hard.
\end{lemma}
\begin{proof}
 Choose $a \in (\Sigma_{1,1} \cup \Sigma_{1,2}) \setminus \Sigma_{2,2}$.
 Then 
 \[ 
  L \cap \Sigma^* a \Sigma^* = \left\{ 
   \begin{array}{ll}
    \Sigma_{1,1}^* a \Sigma_{2,2}^* & \mbox{if } a \in \Sigma_{1,2}; \\
    \Sigma_{1,1}^* a \Sigma_{1,1}^* \Sigma_{1,2} \Sigma_{2,2}^* & \mbox{if } a \in \Sigma_{1,1}.
   \end{array}
  \right.
 \]
 In the first case we can apply Theorem~\ref{thm:uC_PSPACE-hard}
 with $\Gamma = \Sigma_{1,1}, u = a$ and $C = \Sigma_{2,2}$
 to find that $(L\cap \Sigma^* a \Sigma^*)\textsc{-Constr-Sync}$ is $\PSPACE$-hard.
 In the second case, choose some $x \in \Sigma_{1,2}$, then, as, by determinism
 of $\mathcal B$, $x \notin \Sigma_{1,1}$,
 we find $L \cap \Sigma^* ax \Sigma^* = \Sigma_{1,1}^* ax \Sigma_{2,2}^*$
 and we can apply Theorem~\ref{thm:uC_PSPACE-hard}
 with $\Gamma = \Sigma_{1,1}^*$, $u = ax$ and $C = \Sigma_{2,2}$
 to find that $(L \cap \Sigma^* ax \Sigma^*)\textsc{-Constr-Sync}$ is $\PSPACE$-hard.
 Finally, the claim follows by Theorem~\ref{thm:ideal_hardness}. \qed
\end{proof}

The next lemma states a condition such that we get $\PSPACE$-hardness
if the set $\Sigma_{1,1}$ contains at least two distinct symbols.

\begin{lemma}
\label{lem:Sigma11_2_PSPACE-hard}
 Suppose $|\Sigma_{1,1}| \ge 2$, $\Sigma_{1,2}\ne \emptyset$, $\Sigma_{2,1} = \emptyset$
 and $(\Sigma_{1,1} \cup \Sigma_{1,2}) \setminus \Sigma_{2,2} \ne \emptyset$. 
 Then $L(\mathcal B)\textsc{-Constr-Sync}$ is $\PSPACE$-hard.
\end{lemma}
\begin{proof}
 Set $C = \Sigma_{1,1}$ and $\Gamma = \Sigma_{2,2}$.
 By assumption, we find $a \in (\Sigma_{1,1} \cup \Sigma_{1,2}) \setminus \Sigma_{2,2}$.
 If $a \in \Sigma_{1,2}$, then set $u = a$.
 If $a \in \Sigma_{1,1} \setminus \Sigma_{1,2}$, then choose $b \in \Sigma_{1,2}$
 and set $u = ab$.
 Note that, by determinism of the constraint automaton,
 we have $\Sigma_{1,1} \cap \Sigma_{1,2} = \emptyset$.
 Then, $L(\mathcal B) \cap \Sigma^* u \Sigma^* = C^*u\Gamma^*$.
 For this language, the conditions of Theorem~\ref{thm:uC_PSPACE-hard}
 are fulfilled and hence, together with Theorem~\ref{thm:ideal_hardness}, the claim follows.~\qed
\end{proof}

Combining everything, we derive our main result of this section.

\begin{theorem}
\label{thm:two-state-classification}
 For a two-state constraint PDFA $\mathcal B$,
 $L(\mathcal B)\textsc{-Constr-Sync}$
 is $\PSPACE$-complete precisely when
 $\Sigma_{1,2} \ne \emptyset$, $\Sigma_{2,1} = \emptyset$
 and
 $$
  ( \Sigma_{1,1} \cup \Sigma_{1,2} ) \setminus \Sigma_{2,2} \ne \emptyset 
  \mbox{ and } \max\{ |\Sigma_{1,1}|, |\Sigma_{2,2}| \} \ge 2.
 $$
 Otherwise, $L(\mathcal B)\textsc{-Constr-Sync}\in\PTIME$. 
\end{theorem}
\begin{proof}
 We can 
 assume 
 $\Sigma_{1,2} \ne \emptyset$, $\Sigma_{2,1} = \emptyset$
 and $(\Sigma_{1,1} \cup \Sigma_{1,2}) \setminus \Sigma_{2,2}\ne \emptyset$,
 for otherwise, by Proposition~\ref{prop:easy_cases},
 we have $L(\mathcal B)\textsc{-Constr-Sync}\in \PTIME$.
 If $|\Sigma_{1,1}| \ge 2$ or $|\Sigma_{2,2}| \ge 2$,
 by Lemma~\ref{lem:Sigma11_2_PSPACE-hard}
 or Lemma~\ref{lem:Sigma22_2_PSPACE-hard},
 we get $\PSPACE$-hardness, and so,
 by Theorem~\ref{thm:L-contr-sync-PSPACE},
 it is $\PSPACE$-complete in these cases.
 Otherwise, assume $|\Sigma_{1,1}| \le 1$ and $|\Sigma_{2,2}| \le 1$.
 With the other assumptions,
 $$
  L = \bigcup_{x \in \Sigma_{1,2}} \Sigma_{1,1}^* x \Sigma_{2,2}^*.
 $$
 Each language of the form $\Sigma_{1,1}^* x \Sigma_{2,2}^*$
 is over the at most ternary alphabet $\Sigma_{1,1} \cup \{x\}\cup \Sigma_{2,2}$.
 Hence, each such language has 
 the form
 $y^* x z^*$, $xz^*$ or $y^*x$ with $|\{ y,z,x\}| \le 3$ and $\{x,y,z\}\subseteq \Sigma$.
 If a letter is not used in the constraint language, we can, obviously,
 assume the problem is over the smaller alphabet of all letters
 used in the constraint, as usage of letters not occurring in any accepting path in the constraint
 automaton  is forbidden in any input semi-automaton.
 So, by Theorem~\ref{thm:classification_MFCS_paper},
 for the languages $y^* x z^*$ the constraint problem
 is polynomial time solvable, and by Lemma~\ref{lem:union}
 we have $L\textsc{-Constr-Sync}\in \PTIME$. \qed
\end{proof}

\subsection{Three States and Binary Alphabet}
\label{sec:three-states}

Let $\mathcal B = (\Sigma, P, \mu, p_0, F)$
be a three-state constraint PDFA.
Recall the definitions
of the sets $\Sigma_{i,j}$, $1 \le i,j \le 2$ and that here, by our notational conventions,
$P = \{1,2,3\}$, $p_0 = 1$ and $F = \{3\}$. 
First, we will show an analogous result to Lemma~\ref{lem:start-and-final}
for the three-state case, which justifies the mentioned notational conventions.

\begin{lemmarep} 
\label{lem:start-final-three-states}
 Let $\mathcal B = (\Sigma, P, \mu, p_0, F)$ be a PDFA 
 with three states.
 Then, either $L(\mathcal B)\textsc{\--Constr\--Sync}\in \PTIME$,
 or $L(\mathcal B)\textsc{\--Constr\--Sync} \equiv_m^{\log} 
 L(\mathcal B')\textsc{\--Constr\--Sync}$
 for a PDFA $\mathcal B' = (\Sigma, \{1,2,3\}, \mu', 1, \{3\})$.
\end{lemmarep}
\begin{proof}
Let $\mathcal B = (\Sigma, P, \mu, p_0, F)$
be a constraint automaton with $|P| = 3$ and $\emptyset \ne F \subseteq \{1,2,3\}$ arbitrary.
We will assume $P = \{1,2,3\}$ with $p_0 = 1$.
If $F = \{1\}$, then by Theorem~\ref{thm:UVW-theorem},
also the the remark thereafter, we have \[ L(\mathcal B)\textsc{\--Constr\--Sync}\in\PTIME \]
in this case.
So, we can assume $F \ne \{1\}$. 
Also, we can assume that both $2$ and $3$ are reachable from the start state $1$, for non-reachable
states could be removed, giving a constraint automaton with strictly less than
three states, and these
cases where already handled in 
Theorem~\ref{thm:classification_MFCS_paper},
giving polynomial time solvable constraint problems
as $|\Sigma|= 2$.
Set $E = F \setminus \{1\} \ne \emptyset$.
By the previous arguments, $1$ is co-accessible for the final state set $E$,
hence, by Remark~\ref{remark:add-stuff},
$$
 L(\mathcal B_{1, E})\textsc{-Constr-Sync} \equiv_m^{\log}
 L(\mathcal B_{1, F})\textsc{-Constr-Sync}.
$$
So, we can assume $F \subseteq \{2,3\}$.
If $F = \{2,3\}$ and $3$ is reachable from $2$, similarly to the previous
reasoning, then
$$
 L(\mathcal B_{1, \{3\}})\textsc{-Constr-Sync} \equiv_m^{\log}
 L(\mathcal B_{1, \{2,3\}})\textsc{-Constr-Sync}.
$$
Similarly, if $2$ is reachable from $3$. 
Hence, in both cases we can reduce to the case $|F| = 1$.
If none is reachable
from the other, $F = \{2,3\}$ and both are reachable from $1$,
then $\Sigma_{2,1} = \Sigma_{3,1} = \Sigma_{2,3} = \Sigma_{3,2} = \emptyset$ and
$$
 L = \Sigma_{1,1}^*\Sigma_{1,2} \Sigma_{2,2}^* \cup 
     \Sigma_{1,1}^*\Sigma_{1,3} \Sigma_{3,3}^*.
$$
So, $L$ is a union of languages recognizable by two-state automata
over a binary alphabet, both give polynomial time solvable
constraint problems\footnote{Note, for non-binary alphabets, other complexity
classes than $\PTIME$ might arise here.} by Theorem~\ref{thm:classification_MFCS_paper},
hence, by Lemma~\ref{lem:union}, $L\textsc{-Constr-Sync}\in \PTIME$.
So, either the problem is polynomial time solvable,
or equivalent to a constraint automaton whose final state set
is a singleton set non containing the start state,
i.e., without loss of generality we can assume $p_0 = 1$
and $F = \{3\}$ for the state set $P = \{1,2,3\}$ in these cases. 
\qed
\end{proof}

In the general theorem, stated next, the complexity classes we could realize 
depend on the number of strongly connected components
in the constraint automaton.

\begin{theoremrep}
\label{thm:three_states}
 For a constraint PDFA $\mathcal B$ with three states
 over a binary alphabet 
 $L(\mathcal B)\textsc{-Constr-Sync}$
 is either in $\PTIME$, or $\NP\mbox{-complete}$, or $\PSPACE\mbox{-complete}$.
 More specifically, 
 \begin{enumerate}
 \item if $\mathcal B$ is strongly connected
 the problem is always in \PTIME,
 
 \item if the constraint automaton has two strongly connected components,
 the problem is in $\PTIME$ or \PSPACE-complete,
 
 \item and if we have three strongly connected components,
 the problem is either in $\PTIME$ or \NP-complete.
 \end{enumerate}
\end{theoremrep}
\begin{proofsketch}
 This is only a proof sketch,  as even up to symmetry, more than fifty cases have to be checked. We only show a few cases to illustrate how to apply
 the results from Section~\ref{sec:general_results}.
 We will handle the cases illustrated in Table~\ref{tab:three_states_proof_sketch},
 please see the table for the naming of the constraint automata.
 In all automata, the left state is the start state $1$,
 the middle state is state $2$ and the rightmost state
 is state $3$. If not said otherwise, $3$ will be the single final state,
 a convention in correspondence with Lemma~\ref{lem:start-final-three-states}.
 By Theorem~\ref{thm:L-contr-sync-PSPACE},
 for $\PSPACE$-completeness, it is enough to establish
 $\PSPACE$-hardness.
 
 \begin{table}[ht] 
    \centering
 
\begin{tabular}{cll|cll}
  Type & \multicolumn{1}{c}{Automaton} & Complexity & Type & \multicolumn{1}{c}{Automaton} & Complexity \\ \hline
 
$\mathcal B_1$ & \scalebox{.7}{
\begin{tikzpicture}[>=latex',shorten >=1pt,node distance=1.5cm,on grid,auto,baseline=0]
 \tikzset{every state/.style={minimum size=1pt}}
 \node[state] (1) {};
 \node[state] (2) [ right of=1] {};
 \node[state, accepting] (3) [ right of=2] {};
 
 \path[->] (1) edge node {$a$} (2);
 \path[->] (2) edge [bend left] node {$a,b$} (3);
 \path[->] (3) edge [bend left] node {$b$} (2);
\end{tikzpicture}} & $\PSPACE$-c &

$\mathcal B_2$ & \scalebox{.7}{
\begin{tikzpicture}[>=latex',shorten >=1pt,node distance=1.5cm,on grid,auto,baseline=0]
 \tikzset{every state/.style={minimum size=1pt}}
 \node[state] (1) {};
 \node[state] (2) [right of=1] {};
 \node[state, accepting] (3) [right of=2] {};
 
 \path[->] (1) edge node {$a$} (2);
 \path[->] (2) edge [loop above] node {$a$} (2);
 \path[->] (2) edge [bend left] node {$b$} (3);
 \path[->] (3) edge [bend left] node {$b$} (2);
 \path[->] (1) edge [bend right=55] node [above] {$b$} (3);
\end{tikzpicture}} & $\PTIME$ \\ \hline

$\mathcal B_3$ & \scalebox{.7}{
\begin{tikzpicture}[>=latex',shorten >=1pt,node distance=1.5cm,on grid,auto,baseline=0]
 \tikzset{every state/.style={minimum size=1pt}}
 \node[state] (1) {};
 \node[state] (2) [right of=1] {};
 \node[state, accepting] (3) [right of=2] {};
 
 \path[->] (1) edge [loop above] node  {$b$} (1);
 \path[->] (1) edge node {$a$} (2);
 \path[->] (2) edge [loop above] node {$a$} (2);
 \path[->] (2) edge [bend left] node {$b$} (3);
 \path[->] (3) edge [bend left] node {$a$} (2);
\end{tikzpicture}} & $\PSPACE$-c & 

$\mathcal B_{4}$ & \scalebox{.7}{
\begin{tikzpicture}[>=latex',shorten >=1pt,node distance=1.5cm,on grid,auto,baseline=0]
 \tikzset{every state/.style={minimum size=1pt}}
 \node[state] (1) {};
 \node[state] (2) [right of=1] {};
 \node[state, accepting] (3) [right of=2] {};
 
 \path[->] (1) edge [bend right=55] node {$b$} (3);
 \path[->] (1) edge node {$a$} (2);
 \path[->] (2) edge [loop above] node {$b$} (2);
 \path[->] (2) edge [bend left] node {$a$} (3);
 \path[->] (3) edge [bend left] node {$b$} (2);
\end{tikzpicture}} & $\PSPACE$-c \\ \hline

$\mathcal B_{5}$ & \scalebox{.7}{
\begin{tikzpicture}[>=latex',shorten >=1pt,node distance=1.5cm,on grid,auto,baseline=0]
 \tikzset{every state/.style={minimum size=1pt}}
 \node[state] (1) {};
 \node[state] (2) [right of=1] {};
 \node[state, accepting] (3) [right of=2] {};
 
 \path[->] (1) edge node {$a$} (2);
 \path[->] (3) edge [loop above] node {$\Sigma_{3,3}$} (3);
 \path[->] (2) edge [loop above] node {$\Sigma_{2,2}$} (2);
 \path[->] (2) edge [bend left] node {$\Sigma_{2,3}$} (3);
 \path[->] (3) edge [bend left] node {$a \in\Sigma_{3,2}$} (2);
\end{tikzpicture}} & $\PTIME$ &

$\mathcal B_6$ & \scalebox{.7}{
\begin{tikzpicture}[>=latex',shorten >=1pt,node distance=1.5cm,on grid,auto,baseline=0]
 \tikzset{every state/.style={minimum size=1pt}}
 \node[state] (1) {};
 \node[state] (2) [right of=1] {};
 \node[state, accepting] (3) [right of=2] {};
 
 \path[->] (1) edge node {$a$} (2);
 \path[->] (2) edge [loop above] node {$b$} (2);
 \path[->] (2) edge  node {$a$} (3);
 \path[->] (3) edge [loop above] node {$a$} (3);
 \path[->] (1) edge [bend right] node [below] {$b$} (3);
\end{tikzpicture}} & $\NP$-c \\ \hline
\end{tabular}
 \caption{\footnotesize The constraint automata $\mathcal B_i$, $i \in \{1,\ldots,6\}$, with the respective  computational complexities of $L(\mathcal B_i)\textsc{-Constr-Sync}$, for which
  these complexities are proven in the proof sketch of Theorem~\ref{thm:three_states}. Please see
  the main text for more explanation.  }
 \label{tab:three_states_proof_sketch}
 \vspace{-8mm}
\end{table}
 
 \begin{enumerate} 
 \item The constraint automaton\footnote{This constraint automaton was already given in~\cite{DBLP:conf/mfcs/FernauGHHVW19}
 as the single example of a three-state constraint automaton yielding a $\PSPACE$-complete
 problem.} 
 $\mathcal B_1$.
 
   Here $L(\mathcal B_1) = a(a+b)(bb+ba)^*$.
   Set $U = L(\mathcal B_1) \cap \Sigma^* aa \Sigma^* = aa(bb+ba)^*$.
   By Theorem~\ref{thm:ideal_hardness},
   $U\textsc{-Constr-Sync} \le_m^{\log} L(\mathcal B_1)\textsc{-Constr-Sync}$.
   As $(bb+ba) \cap \Sigma^* aa \Sigma^* = \emptyset$
   and $\{bb,ba\}$ is prefix-free, by Theorem~\ref{thm:uC_PSPACE-hard},
   $U\textsc{-Constr-Sync}$ is $\PSPACE$-hard,
   which gives $\PSPACE$-hardness for $L(\mathcal B_1)\textsc{-Constr-Sync}$.
  
 \item The constraint automaton $\mathcal B_2$.
 
   Here $L(\mathcal B_2) = aa^*b(ba^*b)^* \cup b(ba^*b)$.
   We have $aa^*b \subseteq \suff((ba^*b)^*)$
   and $b \subseteq \suff((ba^*b)^*)$.
   By Theorem~\ref{thm:UVW-theorem}
   and Lemma~\ref{lem:union},
   $L(\mathcal B_2)\textsc{-Constr-Sync}\in \PTIME$.
   
 \item The constraint automaton $\mathcal B_3$.
 
   Here $L(\mathcal B_3) = b^*aa^*b(aa^*b)^*$.
   Set $U = L(\mathcal B_3) \cap \Sigma^*bbaba\Sigma^* = b^*bbaba(a+ba)^*b$.
   We have $b^*bbaba(a+ba)^*b \subseteq \factor(b^*bbaba(a+ba)^*)$
   and $b^*bbaba(a+ba)^* \subseteq \factor(b^*bbaba(a+ba)^*b)$.
   Hence, by Theorem~\ref{thm:add-stuff-general}, $U\textsc{-Constr-Sync}$
   has the same computational complexity as synchronization for $b^*bbaba(a+ba)^*$.
   As $(a+ba)^* \cap \Sigma^* bbaba \Sigma^* = \emptyset$
   and $\{a,ba\}$ is a prefix-free set, by Theorem~\ref{thm:uC_PSPACE-hard},
   $(b^*bbaba(a+ba)^*)\textsc{-Constr-Sync}$ is $\PSPACE$-hard,
   and so also synchronization by $U$.
   As, by Theorem~\ref{thm:ideal_hardness},
   $ 
   U\textsc{-Constr-Sync} \le_m^{\log} L(\mathcal B_3)\textsc{-Constr-Sync},
   $
   we get $\PSPACE$-hardness for $L(\mathcal B_3)\textsc{-Constr-Sync}$.
   
 \item The constraint automaton $\mathcal B_4$.

   Here $L(\mathcal B_4) = ab^*a(bb^*a)^* \cup b(bb^*a)^*$.
   Set $U = L(\mathcal B_4) \cap \Sigma^*aab\Sigma^* = aab(b+ab)^*a$.
   As $(b+ab)^* \cap \Sigma^* aab \Sigma^* = \emptyset$
   and $\{b,ab\}$ is a prefix-free set,
   as above, $\PSPACE$-hardness follows by a combination of Theorem~\ref{thm:add-stuff-general}, Theorem~\ref{thm:ideal_hardness}
   and Theorem~\ref{thm:uC_PSPACE-hard}.
   
 \item The constraint automaton $\mathcal B_5$.
 
   Here, $\mathcal B_5$ denotes an entire family of automata.
   In general,
   \[
   L(\mathcal B_5) = a\Sigma_{2,2}^*\Sigma_{2,3}(\Sigma_{3,3}^* \Sigma_{3,2} \Sigma_{2,2}^* \Sigma_{2,3})^*
   \]
   with $a \in \Sigma_{3,2}$.
   As $a \in \Sigma_{3,2}$, we have
   $a\Sigma_{2,2}^*\Sigma_{2,3} \subseteq \Sigma_{3,2} \Sigma_{2,2}^* \Sigma_{2,3}$.
   So, \[ 
   a\Sigma_{2,2}^*\Sigma_{2,3} \subseteq \suff((\Sigma_{3,3}^* \Sigma_{3,2} \Sigma_{2,2}^* \Sigma_{2,3})^*) 
   \]
   and by Theorem~\ref{thm:UVW-theorem} 
   we find $L(\mathcal B_5)\textsc{-Constr-Sync}\in \PTIME$.

 \item The constraint automaton $\mathcal B_6$.
 
  Here, 
  $
   L(\mathcal B_6) = ab^*aa^* \cup ba^*.
  $
  As $L(\mathcal B_6) \subseteq a^* b^* a^* a^*$ 
  the language $L(\mathcal B_6)$
  is a bounded language, hence by Theorem~\ref{thm:bounded_in_NP}
  we have $L(\mathcal B_6)\textsc{-Constr-Sync} \in \NP$.
  Furthermore $L(\mathcal B_6) \cap \Sigma^*abb^*a\Sigma^* = abb^*aa^*$.
  So, by Theorem~\ref{thm:ideal_hardness},
  the original problem is at least as hard as for the constraint
  language $abb^*aa^*$.
  As $ab \notin \factor(b^*)$, $b \notin \factor(aa^*)$
  and $\pref(b^*) \cap aa^* = \emptyset$, by Proposition~\ref{prop:NPc},
  for $abb^*aa^*$ the problem is $\NP$-hard.
  So, by Theorem~\ref{thm:ideal_hardness},
  $L(\mathcal B_6)\textsc{-Constr-Sync}$
  is $\NP$-complete. \qed
 \end{enumerate} 
\end{proofsketch}

\begin{toappendix}
\begin{proof}
 
 For the naming of the states, we will use the same convention
 as written in the proof sketch of Theorem~\ref{thm:three_states}
 in the main text.
 
 \medskip 
 
 \underline{General Assumption:} We will assume $\Sigma_{3,1} = \emptyset$
 and $|\Sigma_{3,3}|\le 1$
 in this proof.
 \begin{quote}
 \emph{Justification of this assumption:} As $L(\mathcal B) = L(\mathcal B_{1,\{3\}}) L(\mathcal B_{3,\{3\}})$,
 if $\Sigma_{3,3} = \{a,b\}$, Theorem~\ref{thm:UVW-theorem}
 would give $L(\mathcal B)\textsc{-Constr-Sync}\in\PTIME$.
 If $\Sigma_{3,1} \ne \emptyset$ and $\Sigma_{2,3} \ne \emptyset$, then $\mathcal B$
 is returning, hence, by Theorem~\ref{thm:returning_poly_alg},
 $L(\mathcal B)\textsc{-Constr-Sync}\in \PTIME$.
 If $\Sigma_{3,1} \ne \emptyset$ and $\Sigma_{2,3} = \emptyset$, 
 then, if $\Sigma_{2,1} = \emptyset$, the state $2$
 is not co-accessible and could be omitted by Lemma~\ref{lem:co-accessible-states-only}.
 Then, $L(\mathcal B)$ could be described by a two-state automaton
 over a binary alphabet, and so, by Theorem~\ref{thm:classification_MFCS_paper},
 $L(\mathcal B)\textsc{-Constr-Sync}\in \PTIME$.
 If, for $\Sigma_{3,1} \ne \emptyset$ and $\Sigma_{2,3} = \emptyset$,
 we have $\Sigma_{2,1} \ne \emptyset$, then $\mathcal B$
 is returning, hence, by Theorem~\ref{thm:returning_poly_alg},
 $L(\mathcal B)\textsc{-Constr-Sync}\in \PTIME$. 
 So, in any case for $\Sigma_{3,1} \ne \emptyset$,
 we find $L(\mathcal B)\textsc{-Constr-Sync}\in\PTIME$,
 and only the cases with $\Sigma_{3,1} = \emptyset$
 remain as interesting cases.
 \end{quote}
 

 \begin{table}[ht] 
    \centering
 
\begin{tabular}{c|c}
 Name & Automaton Template \\ \hline \\
 (a) & \scalebox{.9}{
\begin{tikzpicture}[>=latex',shorten >=1pt,node distance=1.5cm,on grid,auto,baseline=0]
 \tikzset{every state/.style={minimum size=1pt}}

 \node at (3.1,0.25) [rectangle,draw,thick,text width=3.9cm,minimum height=1.05cm,
	text centered,rounded corners, fill=white, name = re] {};
 \node at (3.1,0.5) {Remaining Automaton $\mathcal B^{(i)}$.};
	
 \node[state] (1) {};
 \node[state] (2) [right of=1] {};
 \node[state, accepting] (3) [right of=2] {};
 
 \path[->] (1) edge node {$a$} (2);
 
\end{tikzpicture}} \\ \\

 (b) & \scalebox{.9}{
\begin{tikzpicture}[>=latex',shorten >=1pt,node distance=1.5cm,on grid,auto,baseline=0]
 \tikzset{every state/.style={minimum size=1pt}}

 \node at (3.1,0.25) [rectangle,draw,thick,text width=3.9cm,minimum height=1.05cm,
	text centered,rounded corners, fill=white, name = re] {};
 \node at (3.1,0.5) {Remaining Automaton $\mathcal B^{(i)}$.};
	
 \node[state] (1) {};
 \node[state] (2) [right of=1] {};
 \node[state, accepting] (3) [right of=2] {};
 
 \path[->] (1) edge node {$a$} (2);
 \path[->] (1) edge [loop above] node {$b$} (1);

\end{tikzpicture}} \\ \\

 (c) & \scalebox{.9}{
\begin{tikzpicture}[>=latex',shorten >=1pt,node distance=1.5cm,on grid,auto,baseline=0]
 \tikzset{every state/.style={minimum size=1pt}}

 \node at (3.1,0.25) [rectangle,draw,thick,text width=3.9cm,minimum height=1.05cm,
	text centered,rounded corners, fill=white, name = re] {};
 \node at (3.1,0.5) {Remaining Automaton $\mathcal B^{(i)}$.};
	
 \node[state] (1) {};
 \node[state] (2) [right of=1] {};
 \node[state, accepting] (3) [right of=2] {};
 
 \path[->] (1) edge node {$a,b$} (2);
 
\end{tikzpicture}} \\ \\

 (d) & \scalebox{.9}{
\begin{tikzpicture}[>=latex',shorten >=1pt,node distance=1.5cm,on grid,auto,baseline=0]
 \tikzset{every state/.style={minimum size=1pt}}

 \node at (3.1,0.25) [rectangle,draw,thick,text width=3.9cm,minimum height=1.05cm,
	text centered,rounded corners, fill=white, name = re] {};
 \node at (3.1,0.5) {Remaining Automaton $\mathcal B^{(i)}$.};
	
 \node[state] (1) {};
 \node[state] (2) [right of=1] {};
 \node[state, accepting] (3) [right of=2] {};
 
 \path[->] (1) edge node {$a$} (2);
 \path[->] (1) edge [bend right] node [below] {$b$} (3);

\end{tikzpicture}}
\end{tabular}
 \caption{\footnotesize The automata templates that are combined with
  the partial subautomata $\mathcal B^{(i)}$, $i \in \{1,\ldots,12\}$, from Table~\ref{tab:inner_automata_three_states}
  to form (up to symmetry all relevant) three-state automata. Please see the proof of Theorem~\ref{thm:three_states}
  for explanation.}
 \label{tab:outer_automata_three_states}
  \vspace{-8mm}
\end{table}

\begin{table}[ht] 
    \centering
 
\begin{tabular}{c|c|c|c}
 Name & Two Transitions Above & Name & Two Transitions Below \\ \hline
 
 $\mathcal B^{(1)}$ & \scalebox{.7}{
\begin{tikzpicture}[>=latex',shorten >=1pt,node distance=1.5cm,on grid,auto,baseline=0]
 \tikzset{every state/.style={minimum size=1pt}}
 
 \node[state] (1) {};
 \node[state, accepting] (2) [right of=1] {};
 
 \path[->] (1) edge [loop above] node {$a$} (1);
 \path[->] (1) edge [bend left]  node {$b$} (2);
 \path[->] (2) edge [bend left]  node {$b$} (1);
\end{tikzpicture}} 

 & $\mathcal B^{(7)}$ & 
\scalebox{.7}{
\begin{tikzpicture}[>=latex',shorten >=1pt,node distance=1.5cm,on grid,auto,baseline=0]
 \tikzset{every state/.style={minimum size=1pt}}
 
 \node[state] (1) {};
 \node[state, accepting] (2) [right of=1] {};
 
 \path[->] (2) edge [loop above] node {$a$} (2);
 \path[->] (1) edge [bend left]  node {$a$} (2);
 \path[->] (2) edge [bend left]  node {$b$} (1);
\end{tikzpicture}} 
\\ \hline

 $\mathcal B^{(2)}$ & \scalebox{.7}{
\begin{tikzpicture}[>=latex',shorten >=1pt,node distance=1.5cm,on grid,auto,baseline=0]
 \tikzset{every state/.style={minimum size=1pt}}
 
 \node[state] (1) {};
 \node[state, accepting] (2) [right of=1] {};
 
 \path[->] (1) edge [loop above] node {$a$} (1);
 \path[->] (1) edge [bend left]  node {$b$} (2);
 \path[->] (2) edge [bend left]  node {$a$} (1);
\end{tikzpicture}}

 & $\mathcal B^{(8)}$ & 
\scalebox{.7}{
\begin{tikzpicture}[>=latex',shorten >=1pt,node distance=1.5cm,on grid,auto,baseline=0]
 \tikzset{every state/.style={minimum size=1pt}}
 
 \node[state] (1) {};
 \node[state, accepting] (2) [right of=1] {};
 
 \path[->] (2) edge [loop above] node {$b$} (2);
 \path[->] (1) edge [bend left]  node {$a$} (2);
 \path[->] (2) edge [bend left]  node {$a$} (1);
\end{tikzpicture}} 
\\ \hline

 $\mathcal B^{(3)}$ & \scalebox{.7}{
\begin{tikzpicture}[>=latex',shorten >=1pt,node distance=1.5cm,on grid,auto,baseline=0]
 \tikzset{every state/.style={minimum size=1pt}}
 
 \node[state] (1) {};
 \node[state, accepting] (2) [right of=1] {};
 
 \path[->] (1) edge [loop above] node {$b$} (1);
 \path[->] (1) edge [bend left]  node {$a$} (2);
 \path[->] (2) edge [bend left]  node {$b$} (1);
\end{tikzpicture}} 

 & $\mathcal B^{(9)}$ & 
\scalebox{.7}{
\begin{tikzpicture}[>=latex',shorten >=1pt,node distance=1.5cm,on grid,auto,baseline=0]
 \tikzset{every state/.style={minimum size=1pt}}
 
 \node[state] (1) {};
 \node[state, accepting] (2) [right of=1] {};
 
 \path[->] (2) edge [loop above] node {$a$} (2);
 \path[->] (1) edge [bend left]  node {$b$} (2);
 \path[->] (2) edge [bend left]  node {$b$} (1);
\end{tikzpicture}} 
\\ \hline

 $\mathcal B^{(4)}$ & \scalebox{.7}{
\begin{tikzpicture}[>=latex',shorten >=1pt,node distance=1.5cm,on grid,auto,baseline=0]
 \tikzset{every state/.style={minimum size=1pt}}
 
 \node[state] (1) {};
 \node[state, accepting] (2) [right of=1] {};
 
 \path[->] (1) edge [loop above] node {$b$} (1);
 \path[->] (1) edge [bend left]  node {$a$} (2);
 \path[->] (2) edge [bend left]  node {$a$} (1);
\end{tikzpicture}} 

 & $\mathcal B^{(10)}$ & 
\scalebox{.7}{
\begin{tikzpicture}[>=latex',shorten >=1pt,node distance=1.5cm,on grid,auto,baseline=0]
 \tikzset{every state/.style={minimum size=1pt}}
 
 \node[state] (1) {};
 \node[state, accepting] (2) [right of=1] {};
 
 \path[->] (2) edge [loop above] node {$b$} (2);
 \path[->] (1) edge [bend left]  node {$b$} (2);
 \path[->] (2) edge [bend left]  node {$a$} (1);
\end{tikzpicture}} 
\\ \hline

 $\mathcal B^{(5)}$ & \scalebox{.7}{
\begin{tikzpicture}[>=latex',shorten >=1pt,node distance=1.5cm,on grid,auto,baseline=0]
 \tikzset{every state/.style={minimum size=1pt}}
 
 \node[state] (1) {};
 \node[state, accepting] (2) [right of=1] {};
 
 \path[->] (1) edge [bend left]  node {$a,b$} (2);
 \path[->] (2) edge [bend left]  node {$a$} (1);
\end{tikzpicture}}

 & $\mathcal B^{(11)}$ & 
\scalebox{.7}{
\begin{tikzpicture}[>=latex',shorten >=1pt,node distance=1.5cm,on grid,auto,baseline=0]
 \tikzset{every state/.style={minimum size=1pt}}
 
 \node[state] (1) {};
 \node[state, accepting] (2) [right of=1] {};
 
 \path[->] (1) edge [bend left]  node {$a$} (2);
 \path[->] (2) edge [bend left]  node {$a,b$} (1);
\end{tikzpicture}} 
\\ \hline

 $\mathcal B^{(6)}$ & \scalebox{.7}{
\begin{tikzpicture}[>=latex',shorten >=1pt,node distance=1.5cm,on grid,auto,baseline=0]
 \tikzset{every state/.style={minimum size=1pt}}
 
 \node[state] (1) {};
 \node[state, accepting] (2) [right of=1] {};
 
 \path[->] (1) edge [bend left]  node {$a,b$} (2);
 \path[->] (2) edge [bend left]  node {$b$} (1);
\end{tikzpicture}} 

 & $\mathcal B^{(12)}$ & 
\scalebox{.7}{
\begin{tikzpicture}[>=latex',shorten >=1pt,node distance=1.5cm,on grid,auto,baseline=0]
 \tikzset{every state/.style={minimum size=1pt}}
 
 \node[state] (1) {};
 \node[state, accepting] (2) [right of=1] {};

 \path[->] (1) edge [bend left]  node {$b$} (2);
 \path[->] (2) edge [bend left]  node {$a,b$} (1);
\end{tikzpicture}} \\ \hline

\end{tabular}
 \caption{\footnotesize The partial subautomata between the states $\{2,3\}$
 with precisely three defined transition
 that are combined with the automata templates from Table~\ref{tab:outer_automata_three_states}
 to form a three-state automaton. The cases are sorted such that in the first column, every case
 such that $|\Sigma_{2,2} \cup \Sigma_{2,3}| = 2$ holds is listed, and in the 
 last column every case such that $|\Sigma_{3,3} \cup \Sigma_{3,2}| = 2$,
 where state $2$ is the left state and state $3$ the right state, a naming
 derived from the way how these automata will be combined.
 Please see the proof of Theorem~\ref{thm:three_states}
 for explanation.}
 \label{tab:inner_automata_three_states}
  \vspace{-8mm}
\end{table}

 If $\mathcal B$ consists of a single strongly connected component, then $\mathcal B$
 is returning and the results follows by Theorem~\ref{thm:returning_poly_alg}.
 The remaining cases are handled next.
 
 \medskip 
 
 \noindent i) The case that the states $\{2,3\}$ form one strongly connected component.

 \medskip 
 
 As the final state should be reachable, we can assume $|\Sigma_{1,2} \cup \Sigma_{1,3}| > 0$. 
 Suppose for every letter, we have a transition from the state $2$
 and from the state $3$, i.e., the subautomaton between $2$
 and $3$ is complete and $|\Sigma_{2,2} \cup \Sigma_{2,3}| = |\Sigma_{3,3} \cup \Sigma_{3,2}| = 2$.
 By Remark~\ref{remark:add-stuff}, we can suppose
 both states $2$ and $3$ are final. But then
 $$
  L(\mathcal B) = \Sigma_{1,1}^*\Sigma_{1,2}\Sigma^* \cup \Sigma_{1,1}^*\Sigma_{1,3}\Sigma^*. 
 $$
 Both languages in the union are definable by two-state PDFAs
 over binary alphabet, hence by Theorem~\ref{thm:classification_MFCS_paper}
 give polynomial time solvable constraint problems.
 So, by Lemma~\ref{lem:union}, $L(\mathcal B)\textsc{-Constr-Sync}\in \PTIME$
 in this case.
 So, for the rest of the argument, we only need to handle those cases such that
 for the state $2$ or for the state $3$ at least one outgoing transition 
 for a letter is not defined.
 If at least two transitions are undefined, then, as $\{2,3\}$
 form a strongly connected component and the alphabet is binary,
 precisely two must be defined, one going from state $2$
 to state $3$ and the other back from state $3$ to state $2$.
 In that case
 $$
  L(\mathcal B) = \Sigma_{1,1}^*\Sigma_{1,2}(\Sigma_{2,3}\Sigma_{3,2})^*\Sigma_{2,3}
   \cup \Sigma_{1,1}^*\Sigma_{1,3}(\Sigma_{3,2}\Sigma_{2,3})^*.
 $$
 we have $|\Sigma_{1,1}| \le 1$, as $\Sigma_{1,2} \cup \Sigma_{1,3} \ne \emptyset$. 
 Consider the homomorphism $\varphi : \{a,b,c\} \to \Sigma$
 given by $\varphi(\{a\}) = \Sigma_{1,1}$, $\varphi(\{b\}) = \Sigma_{1,2}$
 and $\varphi(\{c\}) = \Sigma_{3,2}\Sigma_{2,3}$.
 Then, $\varphi(a^*bc^*) = \Sigma_{1,1}^*\Sigma_{1,3}(\Sigma_{3,2}\Sigma_{2,3})^*$,
 and, as $(a^*bc^*)\textsc{-Constr-Sync}\in\PTIME$ by Theorem~\ref{thm:classification_MFCS_paper},
 by Theorem~\ref{thm:hom_lower_bound_complexity}
 we find $(\Sigma_{1,1}^*\Sigma_{1,3}(\Sigma_{3,2}\Sigma_{2,3})^*)\textsc{-Constr-Sync}\in\PTIME$.
 For $\Sigma_{1,1}^*\Sigma_{1,2}(\Sigma_{2,3}\Sigma_{3,2})^*\Sigma_{2,3}$
 we can show, by using Theorem~\ref{thm:add-stuff-general},
 that it has the same computational complexity as
 for the language $\Sigma_{1,1}^*\Sigma_{1,2}(\Sigma_{2,3}\Sigma_{3,2})^*$.
 For the latter language, by a similar argument, using Theorem~\ref{thm:hom_lower_bound_complexity},
 we can show that it gives a polynomial time solvable problem.
 So, with Lemma~\ref{lem:union}, $L(\mathcal B)\textsc{-Constr-Sync}\in\PTIME$.
 Hence, the case that precisely three transitions are defined
 is left.
 In Table~\ref{tab:inner_automata_three_states}, all the possible
 subautomata fulfilling this condition between the states $2$
 and $3$ are listed, or said differently, the table lists all proper partial
 automata with two states and precisely one undefined transition.
 In Table~\ref{tab:outer_automata_three_states},
 all possible ways, up to symmetry, to enter the strongly connected component $\{2,3\}$
 from the start state $1$ are shown.
 Note that indeed every case is covered up to symmetry, as
 we can suppose, without loss of generality, that
 the symbol $a$ goes from state $1$ to state $2$,
 as for the part between $\{2,3\}$ we have \emph{every} possibility, 
 hence in the combinations really get every case such that
 $a \in \Sigma_{1,2}$.
 So, to get all the remaining cases,
 we have to combine each way to enter $\{2,3\}$ from Table~\ref{tab:outer_automata_three_states}
 with every possible strongly component $\{2,3\}$ from 
 Table~\ref{tab:inner_automata_three_states} to cover
 all cases such that in $\{2,3\}$ precisely three transitions are defined.
 In Table~\ref{tab:inner_automata_three_states},
 the cases are sorted such that in the first column, every case
 such that $|\Sigma_{2,2} \cup \Sigma_{2,3}| = 2$ holds is listed, and in the 
 last column every case such that $|\Sigma_{3,3}\cup \Sigma_{3,2}| = 2$. 
 So, we have to check $12 \times 4 = 48$ cases next.
 We will write $\mathcal B$ for the combination
 of some $\mathcal B^{(i)}$, $i \in \{1,\ldots,12\}$,
 with some part from Table~\ref{tab:outer_automata_three_states}
 in each case. First, we select an automaton from Table~\ref{tab:inner_automata_three_states},
 then we investigate each possibility to combine
 it with the first state according to the cases listed
 in Table~\ref{tab:outer_automata_three_states}.

\begin{enumerate} 
\item The inner automaton $\mathcal B^{(1)}$ between the states $\{2,3\}$.

\begin{enumerate}
\item Here $L(\mathcal B) = a(a+bb)^*b$
and $L(\mathcal B_{2,\{2\}})) = (a+bb)^*$.
We have $a \in \suff(L(\mathcal B_{2,\{2\}}))$
and $b \in \pref(L(\mathcal B_{2,\{2\}}))$.
So, by Theorem~\ref{thm:UVW-theorem},
we find $L(\mathcal B)\textsc{-Constr-Sync}\in \PTIME$.

\item  Here $L(\mathcal B) = b^*a(a+bb)^*b$
and $L(\mathcal B_{2,\{2\}})) = (a+bb)^*$.
Similarly as in the previous
case, as $b^*a \subseteq \suff((a+bb)^*)$
and $b \in \pref((a+bb)^*)$,
we find $L(\mathcal B)\textsc{-Constr-Sync}\in \PTIME$.

\item Here $L(\mathcal B) = (a+b)(a+bb)^*b$
 and $L(\mathcal B_{2,\{2\}})) = (a+bb)^*$.
 Similarly, with Theorem~\ref{thm:UVW-theorem}, as in the previous cases, 
 we find $L(\mathcal B)\textsc{-Constr-Sync}\in \PTIME$.

\item Here $L(\mathcal B) = a(a+bb)^*b \cup b \cup bb(a+bb)^*b$.
 
 By Theorem~\ref{thm:UVW-theorem}, similarly as before,
 and Lemma~\ref{lem:finite}, we find that every part of the union
 gives a polynomial time solvable problem.
 Hence, by Lemma~\ref{lem:union}, $L(\mathcal B)\textsc{-Constr-Sync}\in \PTIME$.
 Otherwise, we can note that we
 could simplify $L(\mathcal B) = (a+bb)^*b$
 and using Theorem~\ref{thm:UVW-theorem}.

\end{enumerate}

\item The inner automaton $\mathcal B^{(2)}$ between the states $\{2,3\}$.

\begin{enumerate}
\item Here $L(\mathcal B) = a(a+ba)^*b$.
 Using Theorem~\ref{thm:add-stuff-general},
 we find that 
 $$
  L(\mathcal B)\textsc{-Constr-Sync}\equiv_m^{\log} 
  (a(a+ba)^*)\textsc{-Constr-Sync}.
 $$
 For $a(a+ba)^*$, by Theorem~\ref{thm:UVW-theorem},
 we find $L(\mathcal B)\textsc{-Constr-Sync}\in \PTIME$.

\item Here $L(\mathcal B) = b^*a(a+ba)^*b$.  Using Theorem~\ref{thm:add-stuff-general},
 we find that 
 $$
  L(\mathcal B)\textsc{-Constr-Sync}\equiv_m^{\log} 
  (b^*a(a+ba)^*)\textsc{-Constr-Sync}.
 $$
 Set $U = b^*a(a+ba)^*$.
 We have $U \cap \Sigma^* bba \Sigma^* = b^*bba(a+ba)^*$.
 On the last language, we can apply Theorem~\ref{thm:uC_PSPACE-hard}
 with $\Gamma = \{b\}$, $u = bba$ and $C = \{a,ba\}$
 and find, with Theorem~\ref{thm:ideal_hardness},
 that $U\textsc{-Constr-Sync}$
 is $\PSPACE$-hard.
 So, the original problem $L(\mathcal B)\textsc{-Constr-Sync}$
 is also $\PSPACE$-hard.
 
\item Here $L(\mathcal B) = (a+b)(a+ba)^*b$.
 Set $U = (a+b)(a+ba)^*$.
 By Theorem~\ref{thm:add-stuff-general},
 constrained synchronization for $U$
 has the same computational complexity.
 We find $U \cap \Sigma^*bba\Sigma^* = bba(a+ba)^*$.
 Combining Theorem~\ref{thm:ideal_hardness}
 and Theorem~\ref{thm:uC_PSPACE-hard}
 gives $\PSPACE$-hardness.
 Hence, $L(\mathcal B)\textsc{-Constr-Sync}$
 is $\PSPACE$-hard here.
 
\item Here $L(\mathcal B) = a(a+ba)^*b \cup b \cup ba(a+ba)^*b = (a+ba)^+b \cup b = (a+ba)^*b$.
 By Theorem~\ref{thm:add-stuff-general},
 this language has the same computational complexity
 as synchronization for $(a+ba)^*$.
 The latter language is recognizable by an automaton
 with a single final state that equals the start state.
 So, by Theorem~\ref{thm:UVW-theorem}, see the remark thereafter,
 $(a+ba)^*\textsc{-Constr-Sync}\in \PTIME$.

\end{enumerate}

\item The inner automaton $\mathcal B^{(3)}$ between the states $\{2,3\}$.

\begin{enumerate}
\item Here $L(\mathcal B) = a(b+ab)^*a$.
 By Theorem~\ref{thm:add-stuff-general},
 this gives the same computational complexity
 as $U = a(b+ab)^*$.
 We have $U \cap \Sigma^*aab\Sigma^* = aab(b+ab)^*$,
 hence combining Theorem~\ref{thm:ideal_hardness}
 and Theorem~\ref{thm:uC_PSPACE-hard}
 yields $\PSPACE$-hardness.

\item Here $L(\mathcal B) = b^*a(b+ab)^*a$.
 With $U = b^*a(b+ab)^*$
 and $U \cap \Sigma^*aab\Sigma^* = b^*aab(b+ab)^*$
 we can reason similarly as before
 to find that the problem $L(\mathcal B)\textsc{\--Constr\--Sync}$
 is $\PSPACE$-hard.

\item Here $L(\mathcal B) = (a+b)(b+ab)^*a$.
 The constraint language $U = (a+b)(b+ab)^*$
 has the same complexity by Theorem~\ref{thm:add-stuff-general}.
 Then $U \cap \Sigma^* aab \Sigma^* = aab(b+ab)^*$
 and for this constraint language $\PSPACE$-hardness
 was already shown in case (a) above.

\item Here $L(\mathcal B) = a(b+ab)^*a \cup b \cup bb(b+ab)^*a = (a+bb)(b+ab)^*a \cup b$.
 Then $L(\mathcal B) \cap \Sigma^* aab \Sigma^* =  aab(b+ab)^*a$.
 This constraint language has, by Theorem~\ref{thm:add-stuff-general},
 the same complexity
 as $aab(b+ab)^*$, and for this language
 $\PSPACE$-hardness
 was already shown in case (a) above.
 Hence, with Theorem~\ref{lem:ideal_language},
 we find that $L(\mathcal B)\textsc{-Constr-Sync}$
 is $\PSPACE$-hard here.
\end{enumerate}

\item The inner automaton $\mathcal B^{(4)}$ between the states $\{2,3\}$.

\begin{enumerate}
\item Here $L(\mathcal B) = a(b+aa)^*a$; Theorem~\ref{thm:UVW-theorem}
 gives $L(\mathcal B)\textsc{-Constr-Sync}\in \PTIME$.

\item Here $L(\mathcal B) = b^*a(b+aa)^*a$.
 Set $U = b^*a(b+aa)^*$, which has the same
 complexity by Theorem~\ref{thm:add-stuff-general}.
 As $U \cap \Sigma^* bab\Sigma^* = b^*bab(b+aa)^*$,
 Theorem~\ref{thm:ideal_hardness}
 and Theorem~\ref{thm:uC_PSPACE-hard}
 give $\PSPACE$-hardness.
 
\item Here $L(\mathcal B) = (a+b)(b+aa)^*a$; by Theorem~\ref{thm:UVW-theorem},
 $L(\mathcal B)\textsc{-Constr-Sync}\in \PTIME$.
 
\item Here $L(\mathcal B) = a(b+aa)^*a \cup b \cup ba(b+aa)^*a$.
 Then $L(\mathcal B) \cap \Sigma^*bab\Sigma^* = bab(b+aa)^*a$.
 By a combination of Theorem~\ref{thm:ideal_hardness},
 Theorem~\ref{thm:add-stuff-general}
 and Theorem~\ref{thm:uC_PSPACE-hard},
 we find that the problem is $\PSPACE$-hard.

\end{enumerate}

\item The inner automaton $\mathcal B^{(5)}$ between the states $\{2,3\}$.

\begin{enumerate}
\item Here $L(\mathcal B) = a(a+b)(aa+ab)^*$; by Theorem~\ref{thm:UVW-theorem},
 $L(\mathcal B)\textsc{-Constr-Sync}\in \PTIME$.

\item Here $L(\mathcal B) = b^*a(a+b)(aa+ab)^*$.
 Then $L(\mathcal B) \cap \Sigma^* bbaa \Sigma^* = b^*bbaa(aa+ab)^*$.
 So, by a combination of Theorem~\ref{thm:ideal_hardness}
 and Theorem~\ref{thm:uC_PSPACE-hard}
 we find that the problem is $\PSPACE$-hard.
 
\item Here $L(\mathcal B) = (a+b)(a+b)(aa+ab)^*$
 and $L(\mathcal B) \cap \Sigma^* bb \Sigma^* = bb(aa+ab)^*$.
 Hence, Theorem~\ref{thm:ideal_hardness}
 together with Theorem~\ref{thm:uC_PSPACE-hard},
 with $\Gamma = \emptyset$, $u = bb$ and $C = \{aa,ab\}$,
 we find that the problem is $\PSPACE$-hard.

\item Here $L(\mathcal B) = a(a+b)(aa+ab)^* \cup b(aa+ab)^*$.
 The constraint language $ a(a+b)(aa+ab)^*$
 gives a polynomial time solvable problem
 by Theorem~\ref{thm:UVW-theorem}.
 The constraint language $b(aa+ab)^*$
 has, by Theorem~\ref{thm:add-stuff-general},
 the same complexity as $(aa+ab)^*$.
 The latter gives a polynomial time solvable problem
 by Theorem~\ref{thm:UVW-theorem}, see the remark thereafter.
 So, by Lemma~\ref{lem:union},
 $L(\mathcal B)\textsc{-Constr-Sync}\in \PTIME$.
\end{enumerate}

\item The inner automaton $\mathcal B^{(6)}$ between the states $\{2,3\}$. 

\begin{enumerate}
\item Here $L(\mathcal B) = a(a+b)(ba+bb)^*$.
 We find $L(\mathcal B) \cap \Sigma^*aa\Sigma^* = aa(ba+bb)^*$.
 The latter language yields, by Theorem~\ref{thm:uC_PSPACE-hard},
 a $\PSPACE$-hard problem. So, by Theorem~\ref{thm:ideal_hardness},
 the original problem is $\PSPACE$-hard.

\item Here $L(\mathcal B) = b^*a(a+b)(ba+bb)^*$.
 We find $L(\mathcal B) \cap \Sigma^*aa\Sigma^* = b^*aa(ba+bb)^*$.
 The latter language yields, by Theorem~\ref{thm:uC_PSPACE-hard},
 a $\PSPACE$-hard problem. So, by Theorem~\ref{thm:ideal_hardness},
 the original problem is $\PSPACE$-hard.

\item Here $L(\mathcal B) = (a+b)(a+b)(bb+ba)^*$.
 Then $L(\mathcal B) \cap \Sigma^* aa \Sigma^* = aa(bb+ba)^*$.
 The latter language yields, by Theorem~\ref{thm:uC_PSPACE-hard},
 a $\PSPACE$-hard problem. So, by Theorem~\ref{thm:ideal_hardness},
 the original problem is $\PSPACE$-hard.

\item Here $L(\mathcal B) = a(a+b)(ba+bb)^* \cup b(ba+bb)^*$.
 Then $L(\mathcal B) \cap \Sigma^* aa \Sigma^* = aa(bb+ba)^*$.
 The latter language yields, by Theorem~\ref{thm:uC_PSPACE-hard},
 a $\PSPACE$-hard problem. So, by Theorem~\ref{thm:ideal_hardness},
 the original problem is $\PSPACE$-hard.
\end{enumerate}

\item The inner automaton $\mathcal B^{(7)}$ between the states $\{2,3\}$. 

\begin{enumerate} 
\item Here $L(\mathcal B) = aa(a+ba)^*$; by Theorem~\ref{thm:UVW-theorem},
 $L(\mathcal B)\textsc{-Constr-Sync}\in \PTIME$.

\item Here $L(\mathcal B) = b^*aa(a+ba)^*$. 
 Then $L(\mathcal B) \cap \Sigma^* bbaa \Sigma^* = b^*bbaa(a+ba)^*$.
 Hence, by combining Theorem~\ref{thm:ideal_hardness}
 and Theorem~\ref{thm:uC_PSPACE-hard}, we
 find that $L(\mathcal B)\textsc{-Constr-Sync}$
 is $\PSPACE$-hard.
 
\item Here $L(\mathcal B) = (a+b)a(a+ba)^*$; by Theorem~\ref{thm:UVW-theorem},
 $L(\mathcal B)\textsc{-Constr-Sync}\in \PTIME$.

\item Here $L(\mathcal B) = aa(a+ba)^* \cup b(a+ba)^*$.
 Then $L(\mathcal B) \cap \Sigma^* bb \Sigma^* = bba(a+ba)^*$.
 Hence, by combining Theorem~\ref{thm:ideal_hardness}
 and Theorem~\ref{thm:uC_PSPACE-hard}, we
 find that $L(\mathcal B)\textsc{-Constr-Sync}$
 is $\PSPACE$-hard.
\end{enumerate} 

\item The inner automaton $\mathcal B^{(8)}$ between the states $\{2,3\}$. 

\begin{enumerate}
\item $L(\mathcal B) = aa(b+aa)^*$; by Theorem~\ref{thm:UVW-theorem},
 $L(\mathcal B)\textsc{-Constr-Sync}\in \PTIME$.

\item $L(\mathcal B) = b^*aa(b+aa)^*$; by Theorem~\ref{thm:UVW-theorem},
 $L(\mathcal B)\textsc{-Constr-Sync}\in \PTIME$.

\item $L(\mathcal B) = (a+b)a(b+aa)^*$.
 Then $L(\mathcal B) \cap \Sigma^* bab \Sigma^* = bab(b+aa)^*$.
 Hence, by combining Theorem~\ref{thm:ideal_hardness}
 and Theorem~\ref{thm:uC_PSPACE-hard}, we
 find that the constrained synchronization with $L(\mathcal B)$
 is $\PSPACE$-hard.
 
\item Here $L(\mathcal B) = aa(b+aa)^* \cup b(b+aa)^* = (aa+b)(aa+b)^*$; 
 by Theorem~\ref{thm:UVW-theorem},
 $L(\mathcal B)\textsc{-Constr-Sync}\in \PTIME$.
\end{enumerate}

\item The inner automaton $\mathcal B^{(9)}$ between the states $\{2,3\}$. 

\begin{enumerate}
\item Here $L(\mathcal B) = ab(a+bb)^*$.
 Then $L(\mathcal B) \cap \Sigma^* aba \Sigma^* = aba(a+bb)^*$.
 Hence, by combining Theorem~\ref{thm:ideal_hardness}
 and Theorem~\ref{thm:uC_PSPACE-hard}, we
 find that the constrained synchronization with $L(\mathcal B)$
 is $\PSPACE$-hard.

\item Here $L(\mathcal B) = b^*ab(a+bb)^*$.
 Then $L(\mathcal B) \cap \Sigma^* aba \Sigma^* = b^*aba(a+bb)^*$.
 Hence, by combining Theorem~\ref{thm:ideal_hardness}
 and Theorem~\ref{thm:uC_PSPACE-hard}, we
 find that the constrained synchronization with $L(\mathcal B)$
 is $\PSPACE$-hard.

\item Here $L(\mathcal B) = (a+b)b(a+bb)^*$.
 Then $L(\mathcal B) \cap \Sigma^* aba \Sigma^* = aba(a+bb)^*$.
 Hence, by combining Theorem~\ref{thm:ideal_hardness}
 and Theorem~\ref{thm:uC_PSPACE-hard}, we
 find that the constrained synchronization with $L(\mathcal B)$
 is $\PSPACE$-hard.

\item Here $L(\mathcal B) = ab(a+bb)^* \cup b(a+bb)^*$.
 Then $L(\mathcal B) \cap \Sigma^* aba \Sigma^* = aba(a+bb)^*$.
 Hence, by combining Theorem~\ref{thm:ideal_hardness}
 and Theorem~\ref{thm:uC_PSPACE-hard}, we
 find that the constrained synchronization with $L(\mathcal B)$
 is $\PSPACE$-hard.
\end{enumerate}

\item The inner automaton $\mathcal B^{(10)}$ between the states $\{2,3\}$. 

\begin{enumerate}
\item Here $L(\mathcal B) = ab(b+ab)^*$; by Theorem~\ref{thm:UVW-theorem},
 $L(\mathcal B)\textsc{-Constr-Sync}\in \PTIME$.

\item Here $L(\mathcal B) = b^*ab(b+ab)^*$; by Theorem~\ref{thm:UVW-theorem},
 $L(\mathcal B)\textsc{-Constr-Sync}\in \PTIME$.

\item Here $L(\mathcal B) = (a+b)b(b+ab)^*$; by Theorem~\ref{thm:UVW-theorem},
 $L(\mathcal B)\textsc{-Constr-Sync}\in \PTIME$.
 
\item Here $L(\mathcal B) = (ab+b)(b+ab)^*$; by Theorem~\ref{thm:UVW-theorem},
 $L(\mathcal B)\textsc{-Constr-Sync}\in \PTIME$.
\end{enumerate}

\item The inner automaton $\mathcal B^{(11)}$ between the states $\{2,3\}$. 

\begin{enumerate}
\item Here $L(\mathcal B) = aa(aa+ba)^*$; by Theorem~\ref{thm:UVW-theorem},
 $L(\mathcal B)\textsc{-Constr-Sync}\in \PTIME$.

\item Here $L(\mathcal B) = b^*aa(aa+ba)^*$.
 Then $L(\mathcal B) \cap \Sigma^* bba \Sigma^* = b^*bbaa(aa+ba)^*$.
 Hence, by combining Theorem~\ref{thm:ideal_hardness}
 and Theorem~\ref{thm:uC_PSPACE-hard}, we
 find that the constrained synchronization with $L(\mathcal B)$
 is $\PSPACE$-hard.

\item Here $L(\mathcal B) = (a+b)a(aa+ba)^*$; by Theorem~\ref{thm:UVW-theorem},
 $L(\mathcal B)\textsc{-Constr-Sync}\in \PTIME$.

\item Here $L(\mathcal B) = (aa+b)(aa+ba)^*$.
 Then $L(\mathcal B) \cap \Sigma^* bba \Sigma^* = bba(aa+ba)^*$.
 Hence, by combining Theorem~\ref{thm:ideal_hardness}
 and Theorem~\ref{thm:uC_PSPACE-hard}, we
 find that the constrained synchronization with $L(\mathcal B)$
 is $\PSPACE$-hard.
\end{enumerate}

\item The inner automaton $\mathcal B^{(12)}$ between the states $\{2,3\}$. 

\begin{enumerate}
\item Here $L(\mathcal B) = ab(ab+bb)^*$; by Theorem~\ref{thm:UVW-theorem},
 $L(\mathcal B)\textsc{-Constr-Sync}\in \PTIME$.
 
\item Here $L(\mathcal B) = b^*ab(ab+bb)^*$; by Theorem~\ref{thm:UVW-theorem},
 $L(\mathcal B)\textsc{-Constr-Sync}\in \PTIME$.
 
\item Here $L(\mathcal B) = (a+b)b(ab+bb)^*$; by Theorem~\ref{thm:UVW-theorem},
 $L(\mathcal B)\textsc{-Constr-Sync}\in \PTIME$.

\item Here $L(\mathcal B) = (ab+b)(ab+bb)^*$; by Theorem~\ref{thm:UVW-theorem},
 $L(\mathcal B)\textsc{-Constr-Sync}\in \PTIME$.
\end{enumerate}
\end{enumerate}

\noindent ii) The set $\{1,2\}$ is one strongly connected component.

\medskip 

We will also use the following result here, stating
that, regarding strongly connected components, the placement of the starting
state could be arbitrary.

\begin{quote}
\begin{theorem}
\label{thm:initial_states_in_one_SCC}
 Let $\mathcal B = (\Sigma, P, \mu, p_0, F)$ be a constraint automaton.
 Let $R \subseteq P$ be all states from the same strongly connected
 component as $p_0$, i.e., for each $r \in R$ we have words $u,v \in \Sigma^*$
 such that $\mu(r, u) = p_0$ and $\mu(p_0, v) = r$.
 Then, for any $p \in R$,
 $$
  L(\mathcal B)\textsc{-Constr-Sync} \equiv_m^{\log} L(\mathcal B_{p, F})\textsc{-Constr-Sync}.
 $$
\end{theorem}
\begin{proof}
 Notation as in the statement. Suppose $\mathcal A = (\Sigma, Q, \delta)$ is a semi-automaton.
 Then, $\mathcal A$ has a synchronizing word in $L(\mathcal B)$
 if and only if it has one in $L(\mathcal B_{p, F})$.
 For, if we have $|\delta(Q, w)| = 1$ with $\mu(p_0, w) \in F$,
 then choose $v \in \Sigma^*$ with $\mu(p, v) = p_0$.
 Hence $vw \in L(\mathcal B_{p, F})$ and $|\delta(Q, vw)| = 1$.
 Conversely, if we have $w \in L(\mathcal B_{p, F})$
 with $|\delta(Q, w)|$. Then choose $v \in \Sigma^*$
 with $\mu(p_0, v) = p$, and we have
 $|\delta(Q, vw)| = 1$ and $vw \in L(\mathcal B)$. \qed
\end{proof}
\end{quote}

Recall that by the assumption at the very beginning of this proof $|\Sigma_{3,3}|\le 1$
and $\Sigma_{3,1} = \emptyset$.
If $\Sigma_{1,3} = \Sigma_{2,3} = \emptyset$, we cannot reach the final state,
hence $L(\mathcal B) = \emptyset$.
By assumption, the state set $\{1,2\}$ is one strongly connected component,
which implies $\Sigma_{2,1} \ne \emptyset$.
%
So, if $\Sigma_{1,2} = \{a,b\}$, which implies $\Sigma_{1,3} = \emptyset$,
in the only remaining case, stated in Claim~1 below, 
the problem is a $\PSPACE$-hard problem.

\medskip 

\noindent (a) The case $\Sigma_{1,3} \ne \emptyset$.

\medskip

As $\{1,2\}$ is a strongly connected component,
we must have $\Sigma_{1,2} \ne \emptyset$.
If $\Sigma_{1,3}$ and $\Sigma_{1,2}$
are both non-empty, we must have, by determinism and as $\Sigma = \{a,b\}$,
$\Sigma_{1,1} = \emptyset$.
Assume, without loss of generality,
$\Sigma_{1,2} = \{a\}$, then $\Sigma_{1,3} = \{ b \}$.
Then
\[
L = b\Sigma_{3,3}^* \cup (a\Sigma_{2,1})^*\Sigma_{2,3}\Sigma_{3,3}^*.
\]
The symbols in $\Sigma_{2,1}$ and $\Sigma_{2,3}$
must be different, and both sets are assumed to be non-empty
by the previous arguments.
By Theorem~\ref{thm:classification_MFCS_paper}, for the constraint language $b\Sigma_{3,3}^*$
the constrained synchronization problem is polynomial time solvable. 
For the other language, consider the homomorphism $\varphi : \{a,b,c\}^* \to \{a,b\}^*$
given by $\varphi(\{a\}) = a\Sigma_{2,1}$, $\varphi(\{b\}) = \Sigma_{2,3}$
and $\varphi(\{c\}) = \Sigma_{3,3}$ if $|\Sigma_{3,3}| = 1$,
$\varphi(c) = \varepsilon$ otherwise (note that $|\Sigma| \le 1$).
Then $\varphi(a^*bc^*) = (a\Sigma_{2,1})^*\Sigma_{2,3}\Sigma_{3,3}^*$
and $(a^*bc^*)\textsc{-Constr-Sync}\in\PTIME$
by Theorem~\ref{thm:classification_MFCS_paper}, so that the constraint
problem for $\varphi(a^*bc^*)$ is also polynomial time solvable
by Theorem~\ref{thm:hom_lower_bound_complexity}.
So, with Lemma~\ref{lem:union}, $L\textsc{-Constr-Sync}\in \PTIME$.

\medskip

\noindent (b) The case $\Sigma_{1,3} = \emptyset$, $\Sigma_{1,2} \ne \emptyset$, $\Sigma_{2,1}\ne \emptyset$, $\Sigma_{2,3} \ne \emptyset$ and $|\Sigma_{3,3}|\le 1$.

\medskip 

If $\Sigma_{1,1} \ne \emptyset$, then by Claim~2 below the constraint problem
is $\PSPACE$-hard.
If $\Sigma_{1,1} = \emptyset$ and $|\Sigma_{1,2}| = 1$, we have
$$
 L = (\Sigma_{1,2}\Sigma_{2,1})^*\Sigma_{2,3}\Sigma_{3,3}^*.
$$
As $1 = |\Sigma_{1,2}| = |\Sigma_{2,3}|$, similarly as above
the corresponding constrained problem is polynomial time solvable.


\medskip

\noindent\underline{Claim 1: } If $\Sigma_{1,2} = \{a,b\}$,
$\Sigma_{2,1}$ and $\Sigma_{2,3}$ are non-empty and $|\Sigma_{3,3}| \le 1$,
the problem is $\PSPACE$-hard.
\begin{quote}
 As $\Sigma = \{a,b\}$
 and by determinism of $\mathcal B$, 
 we must have $|\Sigma_{2,1}| = |\Sigma_{2,3}| = 1$
 with a distinct symbol in each set.
 Without loss of generality, we can assume $\Sigma_{2,1} = \{ b \}$
 and $\Sigma_{1,2} = \{a\}$.
 By Lemma~\ref{lem:start-final-three-states},
 we can suppose $\mathcal B = (\Sigma, \{1,2,3\}, \mu, 1, \{3\})$.
 Then, $L(\mathcal B) = (ab+bb)^*(a+b)a\Sigma_{3,3}^*$.
 We have $L(\mathcal B) \cap \Sigma^*bbaa\Sigma^* = (ab+bb)bbaa\Sigma_{3,3}^*$.
 Set $C = \{ab,bb\}$, $u = bbaa$ and $\Gamma = \Sigma_{3,3}$. We have
 \begin{enumerate}
 \item $C$ is finite, prefix-free and $C^* \cap \Sigma^* bbaa \Sigma^* = \emptyset$;
 \item as $|\Gamma| \le 1$, $u$ contains at least one letter not in $\Gamma$;
 \item $\suff(u) \cap \pref(u) = \{ \varepsilon, u \}$;
 \item for $x = bb \in C$, if $vbbw \in \{ab,bb,bbaa\}^*$ with $v,w \in \Sigma^*$,
 then $vbb \in \{ab,bb,bbaa\}^*$.
 \end{enumerate}
 So, we can apply Theorem~\ref{thm:uC_PSPACE-hard}, which gives $\PSPACE$-hardness.
\end{quote}

\noindent\underline{Claim 2:} 
If $\Sigma_{1,1} \ne \emptyset$, $\Sigma_{1,2}\ne \emptyset$,
$\Sigma_{2,1} \ne \emptyset$, $\Sigma_{2,3}\ne \emptyset$
and $|\Sigma_{3,3}| \le 1$, then
$L(\mathcal B)\textsc{-Constr-Sync}$ is $\PSPACE$-hard.

Without loss of generality, assume $\Sigma_{1,1} = \{a\}$, and hence
$\Sigma_{1,2} = \{ b \}$.
We distinguish two cases for the sets $\Sigma_{2,1}$
and $\Sigma_{2,3}$.

\begin{enumerate}
\item $\Sigma_{2,1} = \{a\}$, $\Sigma_{2,3} = \{b\}$, $|\Sigma_{3,3}|\le 1$

 We have $L(\mathcal B) = (a + ba)^*bb\Sigma_{3,3}^*$
 and $L(\mathcal B) \cap \Sigma^* abb \Sigma^* = (a+ba)^*abb\Sigma_{3,3}^*$.
 Set $C = \{a,ba\}$, $u = abb$, $\Gamma = \Sigma_{3,3}$.
 Then,
 \begin{enumerate}
 \item $C$ is finite, prefix-free and $C^* \cap \Sigma^* abb \Sigma^* = \emptyset$;
 \item as $|\Gamma| \le 1$, $u$ contains at least one letter not in $\Gamma$;
 \item $\suff(u) \cap \pref(u) = \{ \varepsilon, u \}$;
 \item for $x = a \in C$, if $vaw \in \{a,ba,abb\}^*$ with $v,w \in \Sigma^*$,
 then $va \in \{a,ba,abb\}^*$, as we can easily check if we write $vaw$
 as a unique product of words from $\{a,ba,abb\}$ and a case distinction which word in the factorisation
 contains the letter $a$ under consideration.
 \end{enumerate}
 
 So, Theorem~\ref{thm:ideal_hardness}
 and Theorem~\ref{thm:uC_PSPACE-hard} give $\PSPACE$-hardness.

\item $\Sigma_{2,1} = \{b\}$, $\Sigma_{2,3} = \{a\}$, $|\Sigma_{3,3}|\le 1$    

 We have $L(\mathcal B) = (a + bb)^*ba\Sigma_{3,3}^*$.
 Hence, $L(\mathcal B) \cap \Sigma^* bbaba \Sigma^* = (a+bb)^*bbaba\Sigma_{3,3}^*$.
 With $C = \{a,bb\}$, $u = bbaba$ and $\Gamma = \Sigma_{3,3}$. We have
 \begin{enumerate}
 \item $C$ is finite, prefix-free and $C^* \cap \Sigma^* bbaba \Sigma^* = \emptyset$;
 \item as $|\Gamma| \le 1$, $u$ contains at least one letter not in $\Gamma$;
 \item $\suff(u) \cap \pref(u) = \{ \varepsilon, u \}$;
 \item for $x = bb \in C$, if $vbbw \in \{a,bb,bbaba\}^*$ with $v,w \in \Sigma^*$,
 then $vbb \in \{a,bb,bbaba\}^*$.
 \end{enumerate}
 So, we can apply Theorem~\ref{thm:uC_PSPACE-hard}, which gives $\PSPACE$-hardness.
\end{enumerate}
Hence, the claim is proven. \qed

\medskip

\noindent iii) The case that $\{1\}, \{2\}$ and $\{3\}$ are all the strongly connected components.

\medskip 

\begin{proposition}\label{prop:add_Sigma11_makes_it_harder}
 Let $\mathcal B = (\Sigma, P, \mu, p_0, F)$
 be some constraint automaton.
 Suppose $\Gamma \subseteq \{ x \mid \{p_0\} \times \{x\} \times P \cap \mu = \emptyset \}$, i.e.,
 we have no transition labelled by letters from $\Gamma$ leaving 
 the start state.
 Then
 $$
  L(\mathcal B)\textsc{-Constr-Sync} \le_m^{\log} (\Gamma^*\cdot L(\mathcal B))\textsc{-Constr-Sync}.
 $$
Intuitively, adding self-loops at the start state gives a harder synchronization problem.
\end{proposition}
\begin{proof}
 We give a reduction from $L(\mathcal B)\textsc{-Constr-Sync}$
 to $(\Gamma^*\cdot L(\mathcal B))\textsc{-Constr-Sync}$.
 Let $\mathcal A = (\Sigma, Q, \delta)$ be some input semi-automaton.
 We can assume $|Q| > 1$, as single state DCSA's are obviously
 synchronizing, for any constraint language.
 We modify $\mathcal A' = (\Sigma, Q', \delta')$ by setting
 $Q' = Q \cup Q_1$, where $Q_1 = \{ q_1 \mid q \in Q \}$
 is a disjoint copy of $Q$. Note the implicit correspondence
 between $Q$ and $Q_1$ set up in the notation. We will use
 this correspondence in the next definition.
 Set
 $$
  \delta'(t, x) = \left\{ 
  \begin{array}{ll}
   t                & \mbox{if } t \in Q_1, x \in \Gamma; \\
   \delta(q, x)     & \mbox{if } t = q_1, q_1 \in Q_1, x \notin \Gamma; \\
   \delta(t,x)      & \mbox{if } t \in Q.
  \end{array}
  \right.
 $$
 Then, $\mathcal A'$ has a synchronizing word in $\Gamma^*\cdot L(\mathcal B)$
 if and only if $\mathcal A$ has a synchronizing word in $L(\mathcal B)$.
 
 \begin{enumerate}
 \item Suppose we have some $w \in L(\mathcal B)$
  with $|\delta(Q, w)| = 1$. As $|Q| > 1$, we must have $|w| > 0$.
  Write $w = xu$ with $x \in \Sigma$. As $w \in L(\mathcal B)$,
  some transition labelled by $x$ must emanate from the
  start state, i.e., $x \notin \Gamma$.
  By construction of $\mathcal A'$, $\delta'(Q_1, x) = \delta'(Q, x)$.
  Hence $\delta'(Q', xu) = \delta'(\delta'(Q,x), u)$.
  As on $Q$, $\mathcal A'$ and $\mathcal A$ operate the same way,
  we have $\delta'(\delta'(Q,x), u) = \delta(Q, w)$.
  Hence, $\mathcal A'$ is synchronized by $w$.
 
 \item Conversely, suppose we have a synchronizing word $w \in \Gamma^*\cdot L(\mathcal B)$
  for $\mathcal A'$.
  As $\delta(Q_1, x) = Q_1$ for any $x \in \Gamma$ and $|Q_1| > 1$,
  in $w$ we must have at least one letter not from $\Gamma$.
  Write $w = uxv$ with $u \in \Gamma^*$, $x \notin \Gamma$
  and $v \in \Sigma^*$.
  Then, by construction, $\delta'(Q', u) = Q_1 \cup \delta(Q, u)$,
  where $\delta'(Q, u) = \delta(Q,u)$, because on $Q$, $\mathcal A'$ and $\mathcal A$ 
  operate the same way.
  As $x \notin \Gamma$, we have $\delta'(Q_1, x) = \delta(Q, x)$.
  Hence $\delta'(Q', ux) = \delta(Q,x) \cup \delta(Q, u)$.
  So 
  $$
   \delta'(Q', uxv) = \delta'(\delta(Q, x) \cup \delta(Q, u), v)
    = \delta(\delta(Q,x) \cup \delta(Q, u), v).
  $$
  As this set is a singleton, and
  $\delta(Q,xv) = \delta(\delta(Q,x),v) \subseteq \delta(\delta(Q,x) \cup \delta(Q, u), v)$,
  the word $xv \in L(\mathcal B)$ synchronizes $\mathcal A$.
 \end{enumerate}
 This shows the statement.~\qed
\end{proof}

 Note that, by assumption, if $\Sigma_{i,j} \ne \emptyset$, $i,j \in \{1,2,3\}$,
 for $i \ne j$, we have $\Sigma_{j,i} = \emptyset$.

\begin{enumerate} 

\item $\Sigma_{1,2} \ne \emptyset$ and $\Sigma_{2,3} = \emptyset$.
 
 Then, as also $\Sigma_{2,1} = \emptyset$,
 the state $2$ is a non-final sink state, which also could
 be omitted by Lemma~\ref{lem:co-accessible-states-only}.
 But then our constraint language could be described
 by a two-state PDFA over a binary alphabet.
 Hence, by Theorem~\ref{thm:classification_MFCS_paper},
 we find $L\textsc{-Constr-Sync}\in\PTIME$.

\item $\Sigma_{1,2} \ne \emptyset$ and $\Sigma_{2,3} \ne \emptyset$.

 Then, $\Sigma_{2,1} = \emptyset$ and $\Sigma_{3,2} = \emptyset$.
 Hence
 $$
  L = \Sigma_{1,1}^* \Sigma_{1,2} \Sigma_{2,2}^* \Sigma_{2,3} \Sigma_{3,3}^*
      \cup \Sigma_{1,1}^* \Sigma_{1,3} \Sigma_{3,3}^*
 $$
 with $\max\{ |\Sigma_{1,1}|, |\Sigma_{2,2}|, |\Sigma_{3,3}|\}\le 1$.
 So $L \subseteq \Sigma_{1,1}^* \Sigma_{1,3}^* \Sigma_{1,2}^* \Sigma_{2,2}^* \Sigma_{2,3}^* \Sigma_{3,3}^*$ and as $|\Sigma_{i,j}| \le 1$ for each $i,j\in \{1,\ldots,3\}$
 we find that $L$ is a bounded language.
 So, by Theorem~\ref{thm:bounded_in_NP}, $L\textsc{-Constr-Sync}\in \NP$.
 We distinguish various sub-cases.

 \begin{enumerate}
 \item[(i)] $\Sigma_{1,2} \ne \emptyset$ and $\Sigma_{1,3} \ne \emptyset$, $\Sigma_{2,3} \ne \emptyset$. 
  As $\Sigma_{1,2} \cap \Sigma_{1,3} = \emptyset$ by determinism
  of $\mathcal B$, we have $\Sigma_{1,1} = \emptyset$.
  Also $|\Sigma_{2,2}| \le 1$ as $\Sigma_{2,3}\ne \emptyset$.
  Recall that by our general assumption, written out at the very beginning
  of the proof, $|\Sigma_{3,3}|\le 1$.
  Write
  $$
   L(\mathcal B) = \Sigma_{1,2}\Sigma_{2,2}^*\Sigma_{3,3}^* \cup \Sigma_{1,3}\Sigma_{3,3}^*.
  $$
  The language $\Sigma_{1,3}\Sigma_{3,3}^*$ 
  could be described
  by a two-state automaton, hence, as we are over a binary alphabet,
  gives a constraint synchronization problem in $\PTIME$
  by Theorem \ref{thm:classification_MFCS_paper}.
  If $\Sigma_{1,2} = \Sigma_{2,2}$, as then
  $$
    \Sigma_{1,2}\Sigma_{2,2}^*\Sigma_{3,3}^* \subseteq \Sigma_{2,2}^*\Sigma_{3,3}^*
     \subseteq \factor(  \Sigma_{1,2}\Sigma_{2,2}^*\Sigma_{3,3}^* ).
  $$
  For $\Sigma_{2,2}^*\Sigma_{3,3}^*$, as $|\Sigma_{3,3}|\le 1$
  and we have a binary alphabet, this language is recognizable by
  a two-state PDFA. Hence, by Theorem~\ref{thm:classification_MFCS_paper},
  gives a polynomial time solvable constraint problem.
  So, by Theorem \ref{thm:add-stuff-general} also the constraint
  synchronization problem for the language $\Sigma_{1,2}\Sigma_{1,2}^*\Sigma_{3,3}^*$
  is in $\PTIME$. With Lemma \ref{lem:union}, then $L(\mathcal B)\textsc{-Constr-Sync} \in \PTIME$.
  Otherwise, if $\Sigma_{1,2} \ne \Sigma_{2,2}$ with $\Sigma_{2,2} \ne \emptyset$, then 
  $\Sigma_{1,2} = \Sigma_{2,3}$.
  We show that, in this case, the problem is $\NP$-complete.
  
  For definiteness, assume $\Sigma_{1,2} = \Sigma_{2,3} = \{a\}$
  and $\Sigma_{2,2} = \{b\}$.

  Then, $\Sigma_{1,1}\cup \Sigma_{1,3} \subseteq \{a\}$, with at least
  one of both sets being empty.
  If $\Sigma_{1,3} = \{a\}$, then, regardless of the value of $\Sigma_{3,3}$,
  we have $L(\mathcal B) \cap \Sigma^* ba^+b \Sigma_{3,3}^* = ba^+b\Sigma_{3,3}^*$.
  Then, if $\Sigma_{3,3} = \{b\}$ or $\Sigma_{3,3} = \emptyset$, apply Proposition~\ref{prop:NPc} 
  with $u = ba$, $v = a$ and $U = b\Sigma_{3,3}^*$.
  The only case left is $ba^+ba^*$, which we handle next.
  
 \begin{claiminproof}
  For $L = ba^+ba^*$, the problem $L\textsc{-Constr-Sync}$
  is $\NP$-hard.
 \end{claiminproof}
 \begin{claimproof}
  By Proposition~\ref{prop:NPc}, for $ba^+b$
  the constrained synchronization problem is
  $\NP$-hard. Looking at the reduction in~\cite{DBLP:conf/ictcs/Hoffmann20},
  it is $\NP$-hard even for input semi-automata with a sink state.
  We will give a reduction from this problem for input semi-automata
  with a sink state to $L\textsc{-Constr-Sync}$.
  Let $\mathcal A = (\{a,b\}, Q, \delta)$ be an input semi-automaton
  with sink state $t \in Q$.
  
  Denote by $t' \notin Q$ a new state.
  Construct $\mathcal A' = (\{a,b\}, Q \times \{1,2\} \cup \{t'\}, \delta')$
  with 
  \begin{align*}
      \delta'((q,1), a) & = (q,1); \\
      \delta'((q,1), b) & = (\delta(q,b), 2); \\ 
      \delta'((q,2), a) & = (\delta(q,a), 2); \\
      \delta'((q,2), b) & = \left\{ 
      \begin{array}{ll}
       (\delta(q,b), 2) & \mbox{if } \delta(q, b) \ne t; \\
       t'               & \mbox{if } \delta(q, b) = t;
      \end{array} 
      \right. \\
      \delta'(t', a) & = \delta'(t', b) = t'.
  \end{align*}
  
  Suppose we have $w \in ba^+b$
  that synchronizes $\mathcal A$.
  Then, as $t$ is a sink state, we have $\delta(Q, w) = \{t\}$.
  Write $w = ba^nb$ with $n > 0$.
  By construction of $\mathcal A'$,
  for any $q \in Q$, we have 
  \begin{align*} 
   \delta'((q, 1), ba^m) & = (\delta(q, ba^m), 2) \\
   \delta'((q, 2), ba^m) & = \left\{
   \begin{array}{ll}
     (\delta(q, ba^m), 2) & \mbox{if } \delta(q, b) \ne t; \\
     t'                   & \mbox{if } \delta(q, b) = t.
   \end{array}
   \right.
  \end{align*}
  As $\delta(t, b) = t$, if $\delta(q, ba^n) = t$,
  we have 
  \[
  \delta'((q, 1), ba^nb) = \delta'((q,2), ba^nb) = t'.
  \]
  If $\delta(q, ba^n) \ne t$, as $\delta(q, b^nb) = t$,
  we also have $\delta'((q, 1), ba^nb) = \delta'((q,2), ba^nb) = t'$.
  So, $\delta'(Q \times \{1,2\} \cup \{t'\}, ba^nb) = \{t'\}$
  and $ba^nb \in ba^+ba^*$.

  Conversely, suppose we have $w \in ba^+ba^*$
  that synchronizes $\mathcal A'$.
  As $t'$ is a sink state in $\mathcal A'$, we have $\delta'(Q \cup \{t'\}, w) = \{t'\}$.
  By construction of $\mathcal A'$, we can only enter $t'$
  from states in $Q \times\{2\}$, and only from those $(q, 2) \in Q\times\{2\}$
  such that $\delta(q, b) = t$ and only by reading the letter $b$.
  So, if we write $w = ba^nba^m$ with $n > 0$, $m \ge 0$, we must have 
  \[
   \delta'(Q\times\{1,2\}\cup \{t'\}, ba^nb) = \{t'\}.
  \] 
  But then, if $q \in Q$, as $\delta'((q,1), ba^nb) = t'$,
  we can deduce $\delta(q, ba^nb) = t$.
  Hence, $\delta(Q, ba^nb) = \{t\}$.
 \end{claimproof}

  If $\Sigma_{1,1} \ne \emptyset$, then 
  we can use Proposition~\ref{prop:add_Sigma11_makes_it_harder}.

  So, for $\Sigma_{1,2} \ne \Sigma_{2,2}$ with $\Sigma_{2,2} \ne \emptyset$,
  under the additional assumptions of this case,
  we have shown that $L(\mathcal B)\textsc{-Constr-Sync}$
  is \NP-complete.
  Now for the next case.
  If $\Sigma_{1,2} \ne \Sigma_{2,2}$ with $\Sigma_{2,2} = \emptyset$,
  then
  $$
   L(\mathcal B_{p_0, \{p_2\}}) = \Sigma_{1,2}\Sigma_{2,3}\Sigma_{3,3}^* \cup \Sigma_{1,3}\Sigma_{3,3}^*. 
  $$
  We have $(\Sigma_{1,3}\Sigma_{3,3}^*)\textsc{-Constr-Sync}\in \PTIME$
  by Theorem \ref{thm:classification_MFCS_paper}, as it could be described by a two-state
  automaton, and we have $|\Sigma| = 2$ here.
  Assume $\Sigma_{1,2} = \{x\}, \Sigma_{2,3}=\{y\}, \Sigma_{3,3}=\{z\}$,
  where we do not suppose these letters to be distinct.
  Define a morphism $\varphi : \{e,f\}^* \to \Sigma^*$
  by $\varphi(e) = xy$ and $\varphi(f) = z$.
  Then $\varphi(ef^*) = \Sigma_{1,2}\Sigma_{2,3}\Sigma_{3,3}^*$.
  The constraint synchronization is in $\PTIME$ for $e^*f$, by Theorem~\ref{thm:classification_MFCS_paper},
  as it is over a binary alphabet and could be described by a two-state automaton.
  Hence, applying Theorem~\ref{thm:hom_lower_bound_complexity}, 
  then gives $(\Sigma_{1,2}\Sigma_{2,}\Sigma_{3,3}^*)\textsc{-Constr-Sync}\in \PTIME$.
  By Lemma~\ref{lem:union}, $L(\mathcal B_{p_0, \{p_2\}})\textsc{-Constr-Sync}\in \PTIME$.

 \item[(ii)] $\Sigma_{1,2} \ne \emptyset$ and $\Sigma_{1,3} = \emptyset$, $\Sigma_{2,3} \ne \emptyset$.
  
  Recall that $|\Sigma_{3,3}|\le 1$.
  That $|\Sigma_{1,1}| \le 1$ and $|\Sigma_{2,2}|\le 1$
  is implied by determinism of $\mathcal B$.

  If $\Sigma_{1,2} = \{a,b\}$ and $\Sigma_{2,2}\ne \emptyset$,
  then $L(\mathcal B)\textsc{-Constr-Sync}$ is $\NP$-complete.
  We have
  $
  L(\mathcal B) = (a+b)\Sigma_{2,2}^*\Sigma_{2,3}^*\Sigma_{3,3}^*.
  $ 
  If $\Sigma_{2,2} = \{a\}$, then $\Sigma_{2,3} = \{b\}$
  and $L(\mathcal B) \cap \Sigma^*ba^+b\Sigma^* = ba^+b\Sigma_{3,3}^*$.
  This case was handled previously as being $\NP$-hard.
  If $\Sigma_{2,2} = \{b\}$, then $\Sigma_{2,3} = \{a\}$. By interchanging $a$
  and $b$ we can argue as in the previous case and find that
  the problem is again $\NP$-hard.

  Now suppose $|\Sigma_{1,2}| = 1$ and $\Sigma_{2,2} \ne \emptyset$.
  Without loss of generality, let $\Sigma_{1,2} = \{b\}$.
  We want to show that this case gives an $\NP$-complete problem
  in the case $\Sigma_{1,2} \ne \Sigma_{2,2}$.
  By Proposition~\ref{prop:add_Sigma11_makes_it_harder}
  it is sufficient if we can establish this for $\Sigma_{1,1} = \emptyset$.
  So, assume $\Sigma_{1,1} = \emptyset$.  
  Then, $L(\mathcal B) = ba^*b\Sigma_{3,3}^* \cap \Sigma^*baab\Sigma^* = ba^+b\Sigma_{3,3}^*$.
  This case was handled previously as being $\NP$-hard.

  If $\Sigma_{1,2} \subseteq \{a,b\}$ and $\Sigma_{2,2} = \emptyset$,
  then
  $$
   L(\mathcal B) = \bigcup_{x\in \Sigma_{1,2}, y\in \Sigma_{2,3}}
     \Sigma_{1,1}^*xy\Sigma_{3,3}^*.
  $$
  Fix $x \in \Sigma_{1,2}$ and $y \in \Sigma_{2,3}$.
  Let $\varphi : \{e, f, g\}^* \to \{a,b\}^*$
  be the homomorphism given by $\varphi(f) = xy$, $\varphi(\{e\}) = \Sigma_{1,2}$
  and $\varphi(\{g\}) = \Sigma_{2,3}$.
  Then, $\varphi(e^*fg^*) = \Sigma_{1,2}^* xy \Sigma_{2,3}^*$.
  By Theorem~\ref{thm:classification_MFCS_paper}, for the constraint $e^*fg^*$,
  the synchronization problem is in $\PTIME$.
  Hence, with Theorem~\ref{thm:hom_lower_bound_complexity},
  \[ 
  (\Sigma_{1,1}^*xy\Sigma_{3,3}^*)\textsc{-Constr-Sync}\in \PTIME
  \]
  so that $L(\mathcal B)\textsc{-Constr-Sync}\in \PTIME$ by Lemma~\ref{lem:union}.
  
  Now consider the case
  $|\Sigma_{1,2}|=1$ and $\Sigma_{2,2}\ne \emptyset$ with $\Sigma_{1,2} = \Sigma_{2,2}$.
  Suppose, for definiteness, $\Sigma_{1,1} = \{b\}$
  and $\Sigma_{1,2} = \Sigma_{2,2} = \{a\}$.
  Then
  $$
   L(\mathcal B) = b^*aa^*b\Sigma_{3,3}^*.
  $$
  
  \begin{claiminproof}
   For $L(\mathcal B) =  b^*aa^*b\Sigma_{3,3}^*$,
   $L(\mathcal B)\textsc{-Constr-Sync}$
   is $\NP$-hard.
  \end{claiminproof}
  \begin{claimproof}
   We show that $L(\mathcal B) \cap \Sigma^* b^+ a^+ b \Sigma^* = b^+a^+b\Sigma_{3,3}^*$
   is $\NP$-hard by giving a reduction from
   the following problem restricted to unary input automata.

\begin{decproblem}\label{def:problem_AutInt}
  \problemtitle{\textsc{Intersection-Non-Emptiness}}
  \probleminput{Deterministic complete automata $\mathcal A_1$, $\mathcal A_2$, \ldots, $\mathcal A_k$.}
  \problemquestion{Is there a word accepted by all of them?}
\end{decproblem}

  This problem is taken from \cite{Koz77} and $\PSPACE$-complete in general, but $\NP$-complete for unary automata, see \cite{fernau2017problems}.
 
   Let $\mathcal A_i = (\{a\}, Q_i, \delta_i, q_i, F_i)$
   be unary automata for $i \in \{1,\ldots,n\}$ 
   with disjoint state sets and $q_i \notin F_i$, which is no essential restriction.
   Construct $\mathcal A = (\{a,b\}, Q_1 \cup \ldots \cup Q_n \cup \{s_1, \ldots, s_n, t\}, \delta)$
   with 
   \begin{align*}
       \delta(q, b) & = \left\{ 
       \begin{array}{ll}
        s_i & \mbox{if } q \in Q_i \setminus F_i; \\ 
        q   & \mbox{if } q \in \{s_1, \ldots, s_n\}; \\
        t   & \mbox{if } q \in F_i \cup \{t\};
       \end{array}
       \right. \\
     \delta(q, a) & = \left\{ 
     \begin{array}{ll}
      \delta_i(q, a) & \mbox{if } q \in Q_i; \\
      q_i            & \mbox{if } \exists i \in \{1,\ldots, n\} : q = s_i; \\
      t              & \mbox{if } q = t.
     \end{array}
     \right.
   \end{align*}
   Then, $\mathcal A$ has a synchronizing word \todo{noch hinschreiben.}
   of the form $b^n a^m bv$ with $n, m > 0$ and $v \in \Sigma_{3,3}^*$
   if and only if $\delta_i(q_i, a^{m-1}) \in F_i$.
  \end{claimproof}
 
 \item[(iii)] $\Sigma_{1,2} = \emptyset$ and $\Sigma_{1,3} \ne \emptyset$, $\Sigma_{3,2} \ne \emptyset$.
 
  By changing the states $2$ and $3$, this is symmetric with case (iii).
  
 \end{enumerate}
\end{enumerate}
 So, we have handled all the possible ways the strongly connected components
 could partition the states and the proof is done. \qed
\end{proof}
\end{toappendix}


\section{Conclusion}
\label{sec:conclusion}

 We have presented general theorems to deduce, for a known constraint language,
 the computational complexity of the corresponding constrained synchronization
 problem. We applied these results to small constraint automata, generalizing
 the classification of two-state automata~\cite{DBLP:conf/mfcs/FernauGHHVW19}
 from an at most ternary alphabet to an arbitrary alphabet.
 We also gave a full classification for three-state constraint automata
 with a binary alphabet.
 Hence, we were able, by using new tools, to strengthen the results from~\cite{DBLP:conf/mfcs/FernauGHHVW19}.
 In light of the methods used and the results obtained so far, it seems
 probable that even for general constraint languages only the three complexity
 classes $\PTIME$, $\PSPACE$-complete or $\NP$-complete arise, hence
 giving a trichotomy result.
 However, we are still far from settling this issue, and much remains to be done
 to answer this question or maybe, surprisingly, present constraint languages
 giving complete problems for other complexity classes.
 Inspection of the results also shows that the $\NP$-complete cases
 are all induced by bounded languages. Hence, the question arises
 if this is always the case, or if we can find non-bounded constraint languages
 giving $\NP$-complete constrained problems.
 
 
 \smallskip \noindent
{\small
\textbf{Acknowledgement}. 
I thank Prof. Dr.
Mikhail V. Volkov for suggesting the problem of constrained synchronization
during the workshop `Modern Complexity Aspects of Formal Languages' that took place  at Trier University  11.--15.\ February, 2019. I also thank anonymous referees of a very preliminary version of this work, whose detailed feedback directly led to this complete reworking of the three-state proof, and anonymous referees of another version.}

\bibliographystyle{splncs04}
\bibliography{msneu} 
\end{document}